\documentclass[11pt]{article}

\usepackage{xcolor}
\usepackage{amsfonts}
\usepackage{amsthm}
\newtheorem{theorem}{Theorem}[section]
\newtheorem{corollary}{Corollary}[section]
\newtheorem{lemma}{Lemma}[section]
\usepackage{graphicx}
\usepackage{algorithm}
\usepackage{algpseudocode}
\usepackage{multirow}
\usepackage{amsmath}

\usepackage{times}
\usepackage{bm}
\usepackage[numbers]{natbib}
\allowdisplaybreaks

\graphicspath{{./art/}}

\DeclareMathOperator{\argmax}{arg\,max}
\newcommand{\floor}[1]{\lfloor #1 \rfloor}

\begin{document}

\title{Likelihood-based Inference for Random Networks with Changepoints}
\author{Daniel Cirkovic, Tiandong Wang, and Xianyang Zhang}
\maketitle

\begin{abstract}
Generative, temporal network models play an important role in analyzing the dependence structure and evolution patterns of complex networks. Due to the complicated nature of real network data, it is often naive to assume that the underlying data-generative mechanism itself is invariant with time. Such observation leads to the study of changepoints or sudden shifts in the distributional structure of the evolving network. In this paper, we propose a likelihood-based methodology to detect changepoints in undirected, affine preferential attachment networks, and establish a hypothesis testing framework to detect a single changepoint, together with a consistent estimator for the changepoint. Such results require establishing consistency and asymptotic normality of the MLE under the changepoint regime, which suffers from long range dependence. The methodology is then extended to the multiple changepoint setting via both a sliding window method and a more computationally efficient score statistic. We also compare the proposed methodology with previously developed non-parametric estimators of the changepoint via simulation, and the methods developed herein are applied to modeling the popularity of a topic in a Twitter network over time. 
\end{abstract}





\section{Introduction}

In network analysis, understanding how a temporal network evolves is a consequential task. For example, assessing how interest in a topic diffuses over time in a Twitter network can only be evaluated if the network is observed temporally. Similarly, in order to judge how effective a policy is at dampening disease spread throughout a community, one must be able to compare the spread dynamics before and after the policy is instituted. Upon observing a large, temporal network, however, it is often na\"ive to assume that it evolved under one fixed data-generating process. Most temporal networks are exposed to forces, internal or external, that may augment the nature in which they evolve. Continuing the Twitter example, news of a major world event may receive increased interest, albeit for a short period of time, before returning to its pre-announcement popularity. These shifts in data-generating dynamics can be characterized as changepoints. 

When it comes to modeling dynamic networks, one popular class of temporal, generative network models is the preferential attachment (PA) model. Since its initial emergence as the Barab\'asi-Albert (BA) model in \cite{barabasi1999emergence}, many generalizations have been developed in order to generate networks that better conform to real-world datasets; see for example \cite{bollobas2003directed, cirkovic2023preferential, gao2017consistent, wang2023poisson}. Estimating changepoints in the undirected preferential attachment model was first studied by \cite{banerjee2023fluctuation} and \cite{bhamidi2018change} where they employed non-parametric estimators of the changepoint based on fluctuations in the degree distribution. Although the methodology is general, it does not directly determine the existence of changepoints. Real network datasets are inherently noisy objects and thus it is natural to ask whether a certain temporal shock has a large enough impact on the network formation, beyond what is typically expected, to be classified as a changepoint. Inspired by \cite{banerjee2023fluctuation} and \cite{bhamidi2018change}, we propose likelihood-based alternatives to the changepoint problem, and provide a hypothesis testing framework that not only detects where a changepoint has occurred, but decides whether or not it even occurred in the first place. If a changepoint has not occurred, one may assume that the entire network is generated under a single PA process. Otherwise, different periods of time may have been governed by varying parameters under the preferential attachment evolution assumption. 

In  \cite{bhamidi2018change}, it is found that the preferential attachment model suffers from long range dependence in the changepoint setting. That is, the degree distribution of a preferential attachment network is heavily influenced by its evolution before changepoint. Since likelihood-based procedures for preferential attachment networks heavily depend on the degree distribution, it could be the case that estimation procedures using edges recently added to the network are corrupted by their dependence on the network's earlier evolution. In this work we confirm that is not the case; despite the long range dependence affecting quantities such as the tail index of the degree distribution, likelihood-based procedures are still effective in making inference on parameters of the preferential attachment model. Surprisingly, there are cases where the long range dependence results in more efficient parameter estimation.

In practice, it is possible to have
multiple changepoints in a given data stream. Although it is justified theoretically, \cite{banerjee2023fluctuation} and \cite{bhamidi2018change} have not addressed the statistical nature of the multiple changepoint detection. Therefore, another goal of our work is to present efficient multiple changepoint detection algorithms that scale to massive networks. Outside of applications to the the PA model, the likelihood-based changepoint detection problem has a long history in the statistics literature. For an overview of both parametric and non-parametric changepoint detection methods, we refer to \cite{csorgo1997limit}. We also refer to \cite{niu2016multiple} for an overview of multiple changepoint procedures. In our present study, we provide two methods, namely sliding window and binary segmentation methods, to handle the multiple changepoint detection problem for PA networks, both of which have also been further applied to a real data example.

The rest of the paper is organized as follows. In Section~\ref{sec:PA}, we provide important background information on the setup of the PA model and the likelihood-based inference. 
Then in Section~\ref{sec:ChangepointEstimation}, we establish the hypothesis testing framework for the single changepoint detection problem in the PA models, together with the proposal of a consistent estimator for the changepoint. We next discuss challenges and our resolutions (i.e. sliding window and binary segmentation methods) for the multiple changepoint situation in Section~\ref{sec:ScoreTest}. A numerical study comparing the two proposed methods is also included in Section~\ref{subsec:multiple_sim}. Section~\ref{sec:data} provides a real data example from a Twitter network, which gives an illustration of the applicability of our proposed methods. We give concluding remarks in Section~\ref{sec:conclusion}, and all technical proofs are collected in the appendix.

\section{Preferential Attachment Modeling}\label{sec:PA}

\subsection{The standard preferential attachment model}

Preferential attachment (PA) has become a popular mechanism to model real-world, dynamic networks by imitating  their scale-free behavior through a ``rich-get-richer" data-generative mechanism. In this section, we review well-known aspects of the undirected preferential attachment model (see Chapter 8 of \cite{van2017random} for a full treatment).

For our theoretical analyses, we assume the graph sequence $\left\lbrace G(k) \right\rbrace_{k = 2}^n$ is initialized with two nodes connected by two edges, i.e. $G(2) \equiv (V(2), E(2))$ where $V(2) = \{ 1, 2 \}$ and $E(2) = \left\lbrace \{1, 2\}, \{1,2 \} \right\rbrace$ .
From $G(k-1)$ to $G(k)$, $k\ge 3$, a new node $k$ together with an undirected edge, $\{k, v_k \}$, is added to the graph according to the following preferential attachment rule. For $v_k \in V(k - 1)$ and $k \geq 3$,
\begin{equation}
\label{PAscheme}
\mathbb{P}\left( \left\lbrace k, v_k  \right\rbrace \in E(k)\setminus E(k - 1) \bigm\vert \mathcal{F}_{k - 1}  \right) = \frac{D_{v_k}(k - 1) + \delta}{(2 + \delta)(k - 1)},
\end{equation}
where $D_{v_k}(k - 1)$ denotes the degree of node $v_k \in V(k - 1)$ and $\mathcal{F}_{k}$ is the sigma algebra generated by the information in $G(k)$. 
Note that we exclude the possibility of having self-loops in the model setup. In order for \eqref{PAscheme} to result in a valid probability, it is further assumed that $\delta > -1$. In what follows, we employ the short-hand notation PA$(\delta)$ to denote this process. 

For $i \geq 1$, let $N_i(k) = \sum_{w \in V(k)} 1_{\{ D_{w}(k) = i \}}$ be the number of nodes with degree $i$ at time $k$. By \cite[Theorem 8.3]{van2017random}, we have 
\begin{equation}
\label{deg_counts}
\frac{N_i(k)}{k} \xrightarrow{p} p_i(\delta) = (2 + \delta)\frac{\Gamma(3 + 2\delta)\Gamma(i + \delta)}{\Gamma(1 + \delta)\Gamma(i + 3 + 2\delta)} \hspace{1cm} \text{as $k \rightarrow \infty$}.
\end{equation}
Also, let $N_{>i}(k) := \sum_{j>i}N_j(k)$ be the number of nodes with degree greater than $i$ at time $k$,
then we have that 
\begin{equation}
\label{deg_counts2}
\frac{N_{>i}(k)}{k} \xrightarrow{p} p_{>i}(\delta) = \frac{\Gamma(3 + 2\delta)\Gamma(i + 1 + \delta)}{\Gamma(1 + \delta)\Gamma(i + 3 + 2\delta)}\hspace{1cm} \text{as $k \rightarrow \infty$}.
\end{equation} 
From \eqref{deg_counts} and \eqref{deg_counts2}, we have the following convenient identity:
\begin{equation}
\label{cdfpmf}
p_{>i}(\delta) = \frac{(i + \delta)p_i(\delta)}{2 + \delta},
\end{equation}
which is helpful later in the proof of our main results. By Stirling's formula, $p_i(\delta) \sim Ci^{-(3 + \delta)}$ as $i \rightarrow \infty$, where $C = (2 + \delta)\Gamma(3 + 2\delta)/\Gamma(1 + \delta)$, giving the PA model the ability to exhibit scale free behavior \citep[see Chapter 8.4 of][for more details]{van2017random}. 

\subsection{Birth-immigration processes}

Results like \eqref{deg_counts} are often attained through either martingale methods or by embedding the degrees $\{D_v(n)\}_{v = 1}^n$ into a carefully constructed sequence of birth-immigration processes. The embedding provides a finer resolution of the structural dynamics that underlie the evolution of the preferential attachment process, as well as a convenient description of limiting objects for preferential attachment models \cite[see][for a full description of such embedddings]{athreya2008growth, rudas2007random}. In order to cleanly describe some limit theorems for the degree distributions, we now introduce linear birth-immigration processes.  

A linear birth process $\left\lbrace \xi^{i}_\theta(t) : t \geq 0 \right\rbrace$ with immigration parameter $\theta$ and unit lifetime parameter is a continuous time Markov process, initialized at $\xi_\theta^i(0) = i$, with state space $\mathbb{N}^+$ and transition rate
\begin{align}
\label{eq:rate}
q_{j, j + 1} = j + \theta, \qquad j \in \mathbb{N}^+, \quad \theta \geq 0.
\end{align}
The difference $\xi^i_\theta(t) - \xi^i_\theta(0)$ has a negative binomial distribution with number of successes $r = i + \theta$ and probability of success $p = e^{-t}$ so that for $t \geq 0$
\begin{align}
\label{eq:pmf}
\mathbb{P}\left( \xi^i_\theta(t) = j \right) =& \frac{\Gamma(j + \theta)}{\Gamma(j - i + 1)\Gamma(i+\theta)}e^{-(i+\theta) t}\left(1 - e^{-t}\right)^{j - i}, \\
\label{eq:cdf}
\mathbb{P}\left( \xi^i_\theta(t) > j \right) =& \frac{\Gamma(j + 1 + \theta)}{\Gamma(j - i + 1)\Gamma(i + \theta)}\int_0^{1 - e^{-t}} x^{j - i}(1-x)^{i - 1 + \theta}dx.
\end{align} 
If the process is initialized with one particle, we drop the superscript and write $\xi_\delta(t) \equiv \xi^1_\delta(t)$ for $t \geq 0$. A simple calculation shows that 
\begin{align}
\label{eq:alt1}
p_i(\delta) &= \mathbb{P}\left( \xi_\delta(T) = i \right), \\
\label{eq:alt2}
p_{>i}(\delta) &= \mathbb{P}\left( \xi_\delta(T) > i \right),
\end{align}
where $T$ is an exponential random variable, independent of $\xi_\delta(\cdot)$, with rate $2 + \delta$.

\subsection{Likelihood inference}\label{sec:like}

In this section, we will review likelihood-based inference for the preferential attachment model. Likelihood inference for the undirected affine preferential attachment model was rigorously studied in \cite{gao2017asymptotic} and will form the basis for our changepoint analyses. Assume we have observed the graph sequence $\left\lbrace G(k) \right\rbrace_{k = 2}^n$ evolving  according to the classical rule in \eqref{PAscheme} with $\delta\in[\eta, K]$, where $\eta > -1$ and $K < \infty$. The likelihood of $\lambda \in [\eta, K]$ given the entire graph sequence is given by
\begin{equation}
\label{likelihood}
L\left(\lambda \mid \left\lbrace G(k) \right\rbrace_{k = 2}^n  \right) = \prod_{k = 3}^n \frac{D_{v_k}(k - 1) + \lambda}{(2 + \lambda)(k - 1)}.
\end{equation}
Since we will later consider model evolution with changepoints, it is convenient to denote 
the likelihood function:
$$
L_{(ns, nt]}(\lambda) := L\left(\lambda \middle\vert \left\lbrace G(k) \right\rbrace_{k = \floor{ns} + 1}^{\floor{nt}}  \right),
\qquad \text{for}\quad s, t \in [0, 1], s<t.
$$ 
If $s$ is such that the likelihood product in \eqref{likelihood} is indexed for $k < 3$, we simply denote the multiples as $1$. The log-likelihood becomes
\begin{align}
&\ell_{(ns, nt]}(\lambda) := \log L_{(ns, nt]}(\lambda) \nonumber\\
&= \sum_{k = \floor{ns} + 1}^{\floor{nt}} \log (D_{v_k}(k - 1) + \lambda) - \sum_{k = \floor{ns} + 1}^{\floor{nt}} \log(k - 1) - (\floor{nt} - \floor{ns})\log(2 + \lambda),\label{eq:loglikelihood}
\end{align}
which emits score and hessian
\begin{align}
u_{(ns, nt]}(\lambda) &:= \frac{\partial}{\partial\lambda} \ell_{(ns, nt]}(\lambda) = \sum_{k = \floor{ns} + 1}^{\floor{nt}} \frac{1}{D_{v_k}(k - 1) + \lambda} - \frac{\floor{nt} - \floor{ns}}{2 + \lambda},\label{score}\\
u'_{(ns, nt]}(\lambda) &:= \frac{\partial^2}{\partial\lambda^2}\ell_{(ns, nt]}(\lambda) = -\sum_{k = \floor{ns} + 1}^{\floor{nt}} \frac{1}{(D_{v_k}(k - 1) + \lambda)^2} + \frac{\floor{nt} - \floor{ns}}{(2 + \lambda)^2}.\label{scoreprime}
\end{align}
Note that if for some $k$, $D_{v_k}(k-1) = i$, node $v_k$ has incremented its degree from $i$ to $i + 1$ at step $k$. Since at most one node, other than node $k$, can have its degree incremented at step $k$, we find that
\begin{align*}
\sum_{k = \floor{ns} + 1}^{\floor{nt}} 1_{\left\lbrace D_{v_k}(k - 1) = i \right\rbrace} = N_{>i}(\floor{nt}) - N_{>i}(\floor{ns}).
\end{align*}
That is, the difference between the number of nodes with degree greater than $i$ at time $\floor{nt}$ and $\floor{ns}$ is simply the number of nodes whose degree have undergone the $i \rightarrow i + 1$ transition during that time period. We may use this fact to rewrite the score and hessian in a more convenient form
\begin{align}
u_{(ns, nt]}(\lambda) = \sum_{i = 1}^\infty \frac{N_{>i}(\floor{nt}) - N_{>i}(\floor{ns})}{i + \lambda}  - \frac{\floor{nt} - \floor{ns}}{2 + \lambda},\label{eq:score2} \\
u'_{(ns, nt]}(\lambda) = -\sum_{i = 1}^\infty \frac{N_{>i}(\floor{nt}) - N_{>i}(\floor{ns})}{(i + \lambda)^2} + \frac{\floor{nt} - \floor{ns}}{(2 + \lambda)^2}, \label{eq:scoreprime2}
\end{align}
(see Lemma \ref{lem:score_prime_asymp} for more details).
We thus define the maximum likelihood estimator based on the edges $\floor{ns} + 1$ to $\floor{nt}$, $\hat{\delta}_{(ns, nt]}$, as the solution to the equation $u_{(ns, nt]}(\delta) = 0$. If we are using the full data to estimate $\delta$, we will employ the notation $\hat{\delta}_n := \hat{\delta}_{(0, n]}$. It is proven in \cite{gao2017asymptotic} that for fixed $s$ and $t$, $\hat{\delta}_{(ns, nt]}$ is a consistent estimator of $\delta$. Additionally, \cite{gao2017asymptotic} prove asymptotic normality of the MLE, where  asymptotic variance is given by $I^{-1}(\delta; \delta)$. Here,
\begin{align}
\label{eq:finull}
I(\lambda; \delta) = \sum_{i = 1}^\infty \frac{p_{>i}(\delta)}{(i + \lambda)^2} - \frac{1}{(2 + \lambda)^2}.
\end{align}
The asymptotic normality of $\hat{\delta}_{(ns, nt]}$ under PA($\delta$) is derived in Lemma~\ref{lem:MLE}, which is crucial in the derivation of our likelihood ratio test.
As \eqref{eq:score2} and \eqref{eq:scoreprime2} suggest, this boils down to understanding the asymptotic behavior of $N_{>i}(\floor{nt})/n$.

\section{Changepoint detection}\label{sec:PAChangepoint}

Although the classical preferential attachment model is already a popular choice for modeling real-world networks, it essentially assumes the network evolution rule (cf. \eqref{PAscheme}) to remain unchanged over time and is not perturbed by temporal shocks to the network. For many real-life scenarios, this is an unrealistic assumption (see for instance the Twitter data example in Section~\ref{sec:data}). In order to capture these dynamics, a preferential attachment model with changepoints have been proposed in \cite{bhamidi2018change}, which we now summarize. The graph sequence $\{ G(k) \}_{k = 2}^n$ evolves as follows:
\begin{itemize}
\item For $k = 2, \dots, \floor{nt^\star}$, allow the network evolve according to PA$(\delta_1)$. 
\item For $k = \floor{nt^\star} + 1, \dots, n$, the network evolves according to PA$(\delta_2)$ where $G(\floor{nt^\star})$ is used as the seed graph.
\end{itemize}
In the rest of this paper, we refer to this graph evolution process as PA$(t^\star; \delta_1, \delta_2)$ with $\delta_1 \neq \delta_2$. The shift from $\delta_1$ to $\delta_2$ thus signifies that some disturbance has occurred at time $\floor{nt^\star} + 1$, and has affected the underlying structure of the network evolution going forward. 

\subsection{Degree distribution under changepoint}

Under the PA$(t^\star; \delta_1, \delta_2)$ model, \cite{bhamidi2018change} derive the asymptotic limit of the empirical degree distribution. This theorem is fundamental in deriving the consistency of our likelihood-based changepoint estimator. We present their result as Theorem \ref{thm:degdistchng} below. 

\begin{theorem}[Theorem 2.1 of \cite{bhamidi2018change}, 3.8 of \cite{banerjee2023fluctuation}]
\label{thm:degdistchng}
For $t^\star \in [0, 1]$, suppose that $\{ G(k) \}_{k = 2}^n$ evolves according to the PA$(t^\star; \delta_1, \delta_2)$ model. Define for $t \geq t^\star$
\begin{align*}
\tau^\star(t) \equiv \frac{1}{2 + \delta_2} \log \frac{t}{t^\star}.
\end{align*}
Then for $i \geq 1$
\begin{align*}
\sup_{t \in [t^\star, 1]} \left| \frac{N_i(\floor{nt})}{nt} -  p^\star_i(t; \delta_1, \delta_2)\right| \xrightarrow{p} 0  \hspace{1cm} \text{as $nk \rightarrow \infty$},
\end{align*}
where
\begin{align}
\label{eq:limchng}
p^\star_i(t; \delta_1, \delta_2) =& \left(1 - (t^\star/t)\right)\mathbb{P}\left(\xi_{\delta_2}(\tilde{T}(t)) = i \right) + (t^\star/t)  \mathbb{P}\left(\xi^{\xi_{\delta_1}(T)}_{\delta_2}(\tau^\star(t)) = i \right).
\end{align}
Here, $T$ is an exponential random variable with rate $2 + \delta_1$, independent of $\xi_{\delta_1}(\cdot)$ and $\xi_{\delta_2}(\cdot)$. $\tilde{T}(t)$ is an exponential random variable, truncated to the interval $[0, \tau^\star(t)]$, with rate $2 + \delta_2$ and independent of $\xi_{\delta_2}(\cdot)$. Explicitly, $\xi^{\xi_{\delta_1}(T)}_{\delta_2}(\cdot)$ is a birth-immigration process generated conditional on the value of $\xi_{\delta_1}(T)$.
\end{theorem}
In order to model  degree growth in the $\text{PA}(t^\star; \delta_1, \delta_2)$ network, \cite{bhamidi2018change} associate to each node a birth-immigration process that models its degree. The constant $\tau^\star(t)$ corresponds to the amount of time the processes evolve after the changepoint, or the amount of embedding time that passes between the introduction of the $\floor{nt^\star}$-th and $\floor{nt}$-th node in the PA$(t^\star; \delta_1, \delta_2)$ process. The random variable $\xi_{\delta_2}(\tilde{T}(t))$ models the degree of a typical node (or birth-immigration process) introduced after the changepoint. Similarly, $\xi^{\xi_{\delta_1}(T)}_{\delta_2}(\tau^\star(t))$ approximates the degree of a typical node introduced before the changepoint. In other words, to generate a random variable that has the asymptotic distribution of a typical node born before the changepoint, we may initialize a birth-immigration process with immigration rate $\delta_2$ with a draw from the distribution $\{p_i(\delta_1)\}_{i \geq 1}$ and allow the birth-immigration process to evolve until time $\tau^\star(t)$.

Using \eqref{eq:pmf} and \eqref{eq:cdf}, we have another convenient representation of \eqref{eq:limchng}
\begin{align*}
p^\star_i(t; \delta_1, \delta_2) =& p_i(\delta_2)\mathbb{P}\left(\xi^3_{2\delta_2}(\tau^\star(t)) > i + 2\right) + (t^\star/t) \sum_{j = 1}^i p_j(\delta_1) \mathbb{P}\left(\xi_{\delta_2}^j(\tau^\star(t)) = i \right).
\end{align*}
Note that in $p^\star_i(t; \delta_1, \delta_2)$, the contribution of $\delta_2$ to the tail has lessend when compared to the PA$(\delta_2)$ model. Since nodes have not undergone their intital evolution under the PA$(\delta_2)$ model, as evidenced by the truncation of the random variable $\tilde{T}(t)$ when compared to \eqref{eq:alt1}, $p_i(\delta_2)$ has been down-weighted in $p^\star_i(t; \delta_1, \delta_2)$. 

From Theorem \ref{thm:degdistchng}, we obtain Lemma \ref{lem:tailchng}.  Recall from the discussion in Section \ref{sec:like} that in order to develop theory for likelihood-based procedures for PA networks, understanding the behavior of $N_{>i}(n)/n$ is pivotal. 
\begin{lemma}
\label{lem:tailchng}
For $t^\star \in [0, 1]$, suppose that $\{ G(k) \}_{k = 2}^n$ evolves according to the PA$(t^\star; \delta_1, \delta_2)$ model. Then for $i \geq 1$
\begin{align*}
\sup_{t \in [t^\star, 1]}\left| \frac{N_{>i}(\floor{nt})}{nt} - p^\star_{>i}(t; \delta_1, \delta_2) \right| \xrightarrow{p} 0  \hspace{1cm} \text{as $n \rightarrow \infty$},
\end{align*}
where
\begin{align*}
p^\star_{>i}(t; \delta_1, \delta_2) =& (1 - (t^\star/t))\mathbb{P}\left(\xi_{\delta_2}(\tilde{T}(t)) > i \right) + (t^\star/t)  \mathbb{P}\left(\xi^{\xi_{\delta_1}(T)}_{\delta_2}(\tau^\star(t)) > i \right) \\
=& p_{>i}(\delta_2)\mathbb{P}\left(\xi^3_{2\delta_2}(\tau^\star(t)) > i + 2\right) - (t^\star/t) \mathbb{P}\left(\xi_{\delta_2}(\tau^\star(t)) > i \right) \\
&+ (t^\star/t) \sum_{j = 1}^i p_j(\delta_1)\mathbb{P}\left( \xi_{\delta_2}^j(\tau^\star(t)) > i\right) + (t^\star/t)  p_{>i}(\delta_1).
\end{align*}
\end{lemma}
The equality for $p^\star_{>i}(t; \delta_1, \delta_2)$ in Lemma \ref{lem:tailchng} is achieved by employing the tower property, conditioning on $\tilde{T}(t)$ and $\xi_{\delta_1}(T)$. We see that preferential attachment models are subject to long range dependence; the initial conditions that govern the early evolution of the network have a long-term impact on properties of the degree distribution. For example, Lemma \ref{lem:tailchng} exhibits that under the PA$(t^\star; \delta_1, \delta_2)$ model, the power-law tail index of the degree distribution is still $2 + \delta_1$. The changepoint does not perturb the tail behavior of the degree distribution which renders power-law tail index estimates such as the Hill estimator unfit to detect shifts in the degree distribution \cite[see][for more on tail index estimation]{hill1975simple,wang2019consistency}. Additionally, it is unclear how such long range dependence affects likelihood inference as presented in Section \ref{sec:like}.

In Lemma \ref{lem:recur}, we develop a generalization of the relationship \eqref{cdfpmf} for the PA$(t^\star; \delta_1, \delta_2)$ model that will prove useful in later proofs.
\begin{lemma}
\label{lem:recur}
Suppose $t > s \geq t^\star$. Then for all $i \geq 1$,
\begin{align*}
\frac{i + \delta_2}{2 + \delta_2} \int_s^t p_i^\star(u; \delta_1, \delta_2) du = tp_{>i}^\star(t; \delta_1, \delta_2) - sp_{>i}^\star(s; \delta_1, \delta_2).
\end{align*}
\end{lemma}

\begin{proof}
Using \eqref{eq:limchng}, we find that
\begin{align*}
\int_s^t p_i^\star(u; \delta_1, \delta_2) du =& p_i(\delta_2)\int_s^t\mathbb{P}\left(\xi^3_{2\delta_2}(\tau^\star(u)) > i + 2\right)du \\
&+ \sum_{j = 1}^i p_j(\delta_1) \int_s^t (t^\star/u) \mathbb{P}\left(\xi_{\delta_2}^j(\tau^\star(u)) = i \right) du,
\end{align*}
where a simple exchange in the order of integration gives that
\begin{align*}
\int_s^t\mathbb{P}\left(\xi^3_{2\delta_2}(\tau^\star(u)) > i + 2\right)du =& t\mathbb{P}\left(\xi^3_{2\delta_2}(\tau^\star(t)) > i + 2\right) - s\mathbb{P}\left(\xi^3_{2\delta_2}(\tau^\star(s)) > i + 2\right) \\
&- t^\star p^{-1}_{>i}(\delta_2) \left(\mathbb{P}\left(\xi_{\delta_2}(\tau^\star(t)) > i \right) - \mathbb{P}\left(\xi_{\delta_2}(\tau^\star(s)) > i \right)\right),
\end{align*}
and a $u$-substitution returns
\begin{align*}
\int_s^t (t^\star/u)  \mathbb{P}\left(\xi_{\delta_2}^j(\tau^\star(u)) = i \right) du =& t^\star\frac{2 + \delta_2}{i + \delta_2}\left(\mathbb{P}\left(\xi^j_{\delta_2}(\tau^\star(t)) > i \right) - \mathbb{P}\left(\xi^j_{\delta_2}(\tau^\star(s)) > i \right) \right).
\end{align*}
Hence
\begin{align*}
\int_s^t p_i^\star(u; \delta_1, \delta_2) du =&   tp_i(\delta_2)\mathbb{P}\left(\xi^3_{2\delta_2}(\tau^\star(t)) > i + 2\right) - sp_i(\delta_2)\mathbb{P}\left(\xi^3_{2\delta_2}(\tau^\star(s)) > i + 2\right) \\
&- t^\star \frac{2 + \delta_2}{i + \delta_2} \left(\mathbb{P}\left(\xi_{\delta_2}(\tau^\star(t)) > i \right) - \mathbb{P}\left(\xi_{\delta_2}(\tau^\star(s)) > i \right)\right) \\
&+  t^\star\frac{2 + \delta_2}{i + \delta_2}\sum_{j = 1}^i p_j(\delta_1)\left(\mathbb{P}\left(\xi^j_{\delta_2}(\tau^\star(t)) > i \right) - \mathbb{P}\left(\xi^j_{\delta_2}(\tau^\star(s)) > i \right) \right),
\end{align*}
where we have used \eqref{cdfpmf} to determine that $p_i(\delta_2)p^{-1}_{>i}(\delta_2) = (2+\delta_2)/(i + \delta_2)$. Multiplying both sides of the previous display by $(i + \delta_2)/(2 + \delta_2)$ and again using \eqref{cdfpmf} completes the proof.
\end{proof}

\subsection{Non-parametric methods}\label{sec:NP}

In order to detect the changepoint $t^\star$ under the PA$(t^\star; \delta_1, \delta_2)$ model, \cite{bhamidi2018change} employs a non-parametric estimator which tracks the proportion of nodes with degree 1 (i.e. leaves) as the network evolves. In a later paper, \citep{banerjee2023fluctuation} develops a non-parametric estimator for $t^\star$ under more general conditions on the attachment rule \eqref{PAscheme}.
As opposed to \cite{bhamidi2018change}, the methodology in \cite{banerjee2023fluctuation} tracks changes in the entire empirical degree distribution, which in-turn provides more information about the changepoint location. Hence, for brevity, we only present and compare our likelihood-based methods to \cite{banerjee2023fluctuation}. Although their theory is based on a preferential attachment model where nodes attach according to their out-degree rather than their total degree, their methodology is still valid for the linear PA case since the associated $\delta$ parameters would just differ by $1$.

The estimator of the changepoint $t^\star$ in \cite{banerjee2023fluctuation} is given by
\begin{equation}
\label{eq:npstat}
\hat{T}_n = \inf \left\lbrace t \geq \frac{1}{h_n} : \sum_{i = 0}^\infty 2^{-i} \left| \frac{N_{i + 1}(\floor{nt})}{nt} - \frac{N_{i + 1}(\floor{n/h_n})}{n/h_n}  \right| > \frac{1}{b_n} \right\rbrace,
\end{equation}
where $h_n$ and $b_n$ are intermediate sequences such that $h_n \rightarrow \infty$, $b_n \rightarrow \infty$, $\frac{h_n}{n} \rightarrow 0$ and $\frac{b_n}{n} \rightarrow 0$ as $n \rightarrow \infty$. Intuitively, \eqref{eq:npstat} detects a changepoint when the $\mathbb{R}_+^\infty$ distance between the empirical degree distributions after the changepoint and some time in the distant past exceeds a given threshold. As reasonable defaults, \cite{banerjee2023fluctuation} recommends using $h_n = \log \log n$ and $b_n = n^{1/ \log \log n}$ which we will also follow. The consistency of $\hat{T}_n$ has been justified in \cite{banerjee2023fluctuation} as well, i.e. $\hat{T}_n \xrightarrow{p} t^\star$, provided that there is a changepoint at $t^\star$. Although methods developed in \citep{banerjee2023fluctuation} are applicable to a wide range of sublinear preferential attachment models, in their current state they do not provide a mechanism to test if a changepoint has occurred in the first place, motivating the need for a hypothesis testing framework and other likelihood-based alternatives.

\subsection{Likelihood ratio test}\label{sec:ChangepointEstimation}

As mentioned earlier, the goal of statistical changepoint detection is two-fold. We need to first decide if a changepoint has occurred, and if so, where it occurred. Both of these objectives can be conveniently met in a hypothesis testing framework. Consider the hypotheses
\begin{align}
\label{eq:hypotheses}
\begin{split}
H_0: \ &\text{The graph sequence  } \{ G(k) \}_{k = 2}^n \text{  evolves according to } \text{PA}(\delta). \\
H_A: \ &\exists t^\star \in [\gamma, 1 - \gamma] \text{ such that } \{ G(k) \}_{k = 2}^n \text{  evolves according to } \text{PA}(t^\star, \delta_1, \delta_2)
\end{split}
\end{align}
for some $\gamma \in (0, 1/2)$. Clearly, the alternative hypothesis implies the existence of a changepoint, while the null hypothesis assumes that the network is generated according to the PA rule \eqref{PAscheme} via a fixed $\delta$. Under $H_A$, $\gamma$ ensures that the changepoint is bounded away from $0$ and $1$. In practice, this assumption guarantees that a sufficient number of edges is generated under both PA schemes (allowing for accurate estimation of $\delta_1$ and $\delta_2$), but also plays a theoretical role in later proofs. Alternatively, one can choose to obviate the need for $\gamma$ by augmenting the statistic used to test the hypotheses \eqref{eq:hypotheses}. See \cite{fryzlewicz2014wild, hariz2007optimal} for more details.

To test \eqref{eq:hypotheses}, we consider the likelihood ratio
\begin{equation}
\label{LR}
\Lambda_m := \frac{L_{(0, n]}(\hat{\delta}_n)}{L_{(0, m]}(\hat{\delta}_{(0, m]}) L_{(m, n]}(\hat{\delta}_{(m, n]})} \qquad \text{for } m = \floor{n\gamma},\dots, \floor{n(1 - \gamma)}.
\end{equation}
The numerator in \eqref{LR} maximizes the likelihood under the null hypothesis of no changepoint and the denominator is maximized under the alternative $\delta_1 \neq \delta_2$ for some conjectured changepoint location $m$. For theoretical analyses, it is instead convenient to work with the statistic $-2 \log \Lambda_m$. With this transformation, a large value of $-2 \log \Lambda_m$ provides evidence in favor of a changepoint located at $m$. Therefore, one may expect the $m$ that maximizes $-2 \log \Lambda_m$ under the alternative to satisfy $m/n\approx t^\star$ as $n \rightarrow \infty$. 

In order to derive the asymptotic distribution for the maximum of $-2 \log \Lambda_m$ under the null hypothesis, it is amenable to convert to continuous time and consider a stochastic process perspective. Consider the log-likelihood ratios $-2 \log\Lambda_{nt}$ indexed by $t \in [\gamma, 1 - \gamma]$. This is a right-continuous process, and hence an element of the Skorohod space $D[\gamma, 1 - \gamma]$ which we will equip with the Skorohod metric. When the limiting process of $-2 \log\Lambda_{ nt}$ is derived, the continuous mapping theorem gives the asymptotic distribution of the supremum of $-2 \log\Lambda_{ nt}$ for $t \in [\gamma, 1 - \gamma]$. 

The following theorem gives the asymptotic distribution of 
$$\sup_{t \in [\gamma, 1 - \gamma]} -2 \log \Lambda_{nt},$$ 
under $H_0$, from which we derive the hypothesis test. Here, we use $B(t)$ to refer to a Brownian bridge process.

\begin{theorem}
\label{thm:nullhyp}
Fix $\gamma \in (0, 1/2)$. Then under $H_0$ in \eqref{eq:hypotheses}
\begin{equation}
\sup_{t \in [\gamma, 1 - \gamma]} -2 \log \Lambda_{nt} \Rightarrow \sup_{t \in [\gamma, 1 - \gamma]} \frac{B^2(t)}{t(1 - t)},
\end{equation}
in $\mathbb{R}$, where $\Rightarrow$ denotes the weak convergence.
\end{theorem}

We defer the detailed proof of Theorem~\ref{thm:nullhyp} to Appendix~\ref{subsec:pf_thm1}.
In fact, Theorem \ref{thm:nullhyp} establishes a hypothesis testing framework for the null of no changepoint, rejects $H_0$ when the test statistics 
$$\sup_{t \in [\gamma, 1 - \gamma]} -2 \log \Lambda_{nt}$$ 
exceeds the $(1 - \alpha)$-th quantile of the limiting distribution, 
$$\sup_{t \in [\gamma, 1 - \gamma]} \frac{B^2(t)}{t(1 - t)}.$$ 
Though quantiles of the limit distribution are not readily available in most software, they can be simulated via Brownian motion realizations or approximations like that of Equation (1.3.26) in \cite{csorgo1997limit}.

Now that Theorem \ref{thm:nullhyp} gives us the ability to detect the presence of a changepoint, we would like to consistently estimate it. As mentioned earlier, large values of $-2 \log \Lambda_{nt}$ indicate that $t$ is likely a changepoint. Hence, under $H_A$, we expect the argmax of $-2 \log \Lambda_{nt}$ to be a consistent estimator of the true changepoint $t^\star$. However, in order for  $-2 \log \Lambda_{nt}$ to attain large values in a neighborhood about $t^\star$, we require that $L_{(nt, n]}(\hat{\delta}_{(nt, n]})$ is large for $t \geq t^\star$. That is, we require that the maximum likelihood estimator $\hat{\delta}_{(nt, n]}$ is a consistent estimator of $\delta_2$ under the PA$(t^\star; \delta_1, \delta_2)$. This condition is highly non-obvious as the long range dependence implied by Theorem \ref{thm:degdistchng} would suggest that the MLE may be corrupted by the pre-changepoint regime. This, however, is not the case as we present in Theorem \ref{thm:delta_unif_chng}, which is proved in Section \ref{sec:delta_unif_chng_proof} in the Appendix.

\begin{theorem}
\label{thm:delta_unif_chng}
Fix $s \in [t^\star, 1]$ and $\tau \in (0, 1 - s)$. Suppose that $\{ G(k) \}_{k = 2}^n$ evolves according to the PA$(t^\star; \delta_1, \delta_2)$ rule.
Then as $n \rightarrow \infty$
\begin{align*}
\sup_{t \in [s + \tau, 1]} \left| \hat{\delta}_{(ns, nt]} - \delta_2 \right| \xrightarrow{p} 0.
\end{align*}
\end{theorem}

We may also prove asymptotic normality of the MLE for $\delta_2$ under the PA$(t^\star; \delta_1, \delta_2)$ model, as given in Theorem \ref{thm:MLE_chng}. Note that the asymptotic variance of the centered and scaled $\hat{\delta}_{(ns, nt]}$ is given by $\left(I^\star_t(\delta_2) - I^\star_s(\delta_2)\right)^{-1}$ when $s \geq t^\star$. Here, 
\begin{align*}
I^\star_t(\lambda) =  t \left(\sum_{i = 1}^\infty \frac{p_{>i}^\star(t; \delta_1, \delta_2)}{(i + \lambda)^2} - \frac{1}{(2 + \lambda)^2}\right).
\end{align*}
Theorem \ref{thm:MLE_chng} indicates that, equating the number of edges used in each estimation procedure, the asymptotic variance of the MLE for $\delta_2$ under the PA$(t^\star; \delta_1, \delta_2)$ regime is smaller than the asymptotic variance of the MLE for $\delta_2$ under the PA$(\delta_2)$ model if $\delta_1 < \delta_2$. On the other hand, if $\delta_1 > \delta_2$, $\hat{\delta}_{(ns, nt]}$ is subject to a larger asymptotic variance under the PA$(t^\star; \delta_1, \delta_2)$ model when compared to the PA$(\delta_2)$ setting. 

This finding can be explained by relationship between the asymptotic variance of the MLE and the tail behavior of the degree distribution under the PA$(t^\star; \delta_1, \delta_2)$ model. Note that $N_{>i}(\floor{nt}) - N_{>i}(\floor{ns})$  counts the number of nodes that transitioned from degree $i$ to degree $i + 1$ in between time steps $\floor{ns}$ and $\floor{nt}$. Thus, if the parameters of the PA$(t^\star; \delta_1, \delta_2)$ model are such that more nodes make transitions between large degrees, $tp_{>i}(t; \delta_1, \delta_2) - sp_{>i}(s; \delta_1, \delta_2)$ tends to decay more slowly as $i \rightarrow \infty$, resulting in a smaller asymptotic variance of the MLE. However, in order for a node to make transitions between large degrees, the node must attain a large degree in the first place. Since $p_i(\delta)$ has a power law tail index of $3 + \delta$, $\delta_1 < \delta_2$ encourages more nodes to attain a large degree before the changepoint than otherwise would have occurred if the network was generated under the PA($\delta_2$) model. The setting where $\delta_1 < \delta_2$ thus acts as a stimulus that fosters the opportunity for more nodes to pass through larger degrees after the changepoint, while  $\delta_1 > \delta_2$ has the opposite effect.

\begin{theorem}\label{thm:MLE_chng}
Fix $s \in [t^\star, 1]$ and $\tau \in (0, 1 - s)$. Assume that $\{ G(k) \}_{k = 2}^n$ evolves according to the PA$(t^\star; \delta_1, \delta_2)$ rule.
Then as $n \rightarrow \infty$
\[ 
\left(I^\star_t(\delta_2) - I^\star_s(\delta_2)\right) \cdot \sqrt{n}(\hat{\delta}_{(ns, nt]} - \delta_2) \Rightarrow
W\left(I^\star_t(\delta_2)\right) - W\left(I^\star_s(\delta_2) \right)
\] in $D[s + \tau, 1]$ where $W(\cdot)$ is a Wiener process. Moreover, if $\delta_1 < \delta_2$,
\begin{align*}
(I_t^\star(\delta_2)& - I_s^\star(\delta_2))^{-1} < (t - s)^{-1} I^{-1}(\delta_2; \delta_2),
\end{align*}
and if $\delta_1 > \delta_2$
\begin{align*}
(I_t^\star(\delta_2)& - I_s^\star(\delta_2))^{-1} > (t - s)^{-1}I^{-1}(\delta_2; \delta_2).
\end{align*}
\end{theorem}
Theorem \ref{thm:MLE_chng} is proved in Section \ref{sec:MLE_chng_proof} of the Appendix. Theorem \ref{thm:consistency} formalizes the consistency of the likelihood ratio procedure. 

\begin{theorem}
\label{thm:consistency}
Fix $\gamma \in (0, 1/2)$ and assume that there exists only one changepoint at $t^\star \in [\gamma, 1 - \gamma]$. Suppose also that the network evolves according to PA$(t^\star; \delta_1, \delta_2)$. Then we have
\begin{equation}\label{eq:est}
\hat{t}_n := n^{-1}\underset{\floor{n\gamma}\le m\le \floor{(1 - \gamma)n} }{\argmax} -2 \log \Lambda_m \xrightarrow{p} t^\star,
\end{equation}
as $n \rightarrow \infty$.
\end{theorem}

We leave the proof of Theorem~\ref{thm:consistency} to Appendix~\ref{subsec:pf_consistency}, and provide comparisons on the performance of $\hat{t}_n$ and $\hat{T}_n$ through a simulation study in the next section.

\subsection{Simulation study}

Though the consistency of $\hat{t}_n$ and $\hat{T}_n$ has been justified, 
we now further evaluate and compare their numerical performance
on simulated data. 
In each simulation, the preferential attachment networks are initialized by a node with a self-loop, and all simulations assume a typical significance level of $\alpha = 0.05$. We first assess the Type 1 error rate for the likelihood ratio methodology by simulating $500$ PA$(\delta)$ networks with $50{,}000$ edges, $\delta \in \{-0.5, 0, 1, 2 \}$ and $\gamma = 0.1$. The Type 1 error rates and their Bernoulli-based standard errors are presented in Table \ref{tab:H0}. As expected, the Wald-based confidence intervals for the rejection rates all contain $\alpha = 0.05$. Hence, the asymptotic distribution in Theorem \ref{thm:nullhyp} is a good approximation for the distribution of the likelihood ratio under the null hypothesis, even at small sample sizes. Most real-world dynamic networks (e.g. those listed on data repositories like SNAP \cite{snapnets} and KONECT \cite{konect}) contain well over $50{,}000$ edges and hence asymptotic approximations are reliable, assuming the network is truly generated by a PA process.

\begin{table}
\centering
\begin{tabular}{ c c c c c}
\hline \hline
$\delta$ & $-0.5$ & $0$ & $1$ & $2$ \\ 
\hline
Rejection Rate & $0.050$ & $0.042$ & $0.046$ & $0.052$  \\  
Standard Error & $0.0097$ & $0.0090$ & $0.0094$ & $0.0099$ \\
\hline 
\end{tabular}
\caption{Rejection rates and (Bernoulli-based) standard errors for the likelihood ratio test applied to 500 PA$(\delta)$ networks with $50{,}000$ edges and $\gamma = 0.1$.}
\label{tab:H0}
\end{table}

Next, we evaluate the statistical power and estimation of the changepoint location for the likelihood ratio methodology. We apply the likelihood ratio test to 500 simulated PA$(0.6; 0,\delta_2)$ networks, i.e. they are initialized by a Barab\'asi-Albert process, PA$(0)$, and then transitions to a PA$(\delta_2)$ process at $t^\star = 0.6$. The networks have $50{,}000$ edges and we again let $\gamma = 0.1$. Here, we allow $\delta_2$ to vary between $-0.2$ and $0.2$ by increments of $0.02$. An empirical power curve and mean absolute error for the changepoint location for tests that rejected the null hypothesis are presented in Figure \ref{fig:HA}. The likelihood ratio methodology is powerful; differences in the offset parameter beyond $0.1$ in absolute value are detected well-over 95\% of the time. Interestingly, decreases in the offset parameter are detected slightly more often than increases. For differences in the offset parameter below $0.1$ in absolute value, however, the estimation of the changepoint location becomes unreliable. 

In order to further assess the accuracy of the changepoint location estimate, we compare the estimation error for the likelihood ratio and non-parametric estimators, $\hat{t}_n$ and $\hat{T}_n$, on larger networks with greater, more realistic differences in the offset parameter. Specifically, we simulate 500 networks of size $100{,}000$ from a PA$(0.6; 0, \delta_2)$ process, where $\delta_2$ is allowed to range from $0.1$ to $0.5$ by increments of $0.1$. The mean absolute error in changepoint location is presented in Table \ref{tab:HA2}. The likelihood ratio test rejects the null hypothesis in all simulations. Compared to previous simulations, $\hat{t}_n$ becomes highly accurate with larger differences in the offset parameter and has uniformly smaller error than $\hat{T}_n$. The small differences in the offset parameter require the addition of many nodes to manifest in the empirical degree distribution according to the law of large numbers effect, hence $\hat{T}_n$ performs poorly for small $\delta_2$. 
We also remark here that the supremacy of the likelihood ratio is unsurprising, since the likelihood-based methodology should out-perform non-parametric procedures when the data-generating model is correctly specified. 

The non-parametric estimator $\hat{T}_n$, however, was designed to accommodate more general preferential attachment models. Namely, $\hat{T}_n$ can detect differences between sublinear regimes of preferential attachment as explained in Assumption 2.1 of \cite{banerjee2023fluctuation}. In order to assess the robustness of the likelihood ratio methodology to more general changes in the degree-based attractiveness of nodes, we generate 500 networks with $100{,}000$ edges, where the attachment function changes from linear to sublinear at $t^\star = 0.6$. 
At step $k$, node $v$ is chosen to be attached to with probability proportional to $D_v(k-1) + 1$ for $k\le \lfloor nt^\star\rfloor$, and $(D_v(k-1) + 1)^b$ for $k\ge \lfloor nt^\star\rfloor+1$. Here we allow $b\in \{0.5, 0.6, 0.7, 0.8, 0.9\}$, where larger values of $b$ thus indicate less dramatic changes in the attachment function. Mean absolute errors (MAE) in changepoint location for $\hat{t}_n$ and $\hat{T}_n$ are reported in Table \ref{tab:HAsub}. Although $\hat{\delta}_{(nt, n]}$ is generally computed under a misspecified model, the likelihood ratio methodology still performs well by virtue of $\hat{\delta}_{(0, nt]}$ and $\hat{\delta}_{(nt, n]}$ being far apart near the true changepoint. The non-parametric estimator $\hat{T}_n$ performs well for lower values of $b$, but if the change in attachment function is subtle, it performs worse. This simulation, along with the previous, indicates that the empirical degree distribution is not sensitive to small changes in the attachment function from a statistical perspective.  

\begin{figure}
\includegraphics[width=\textwidth]{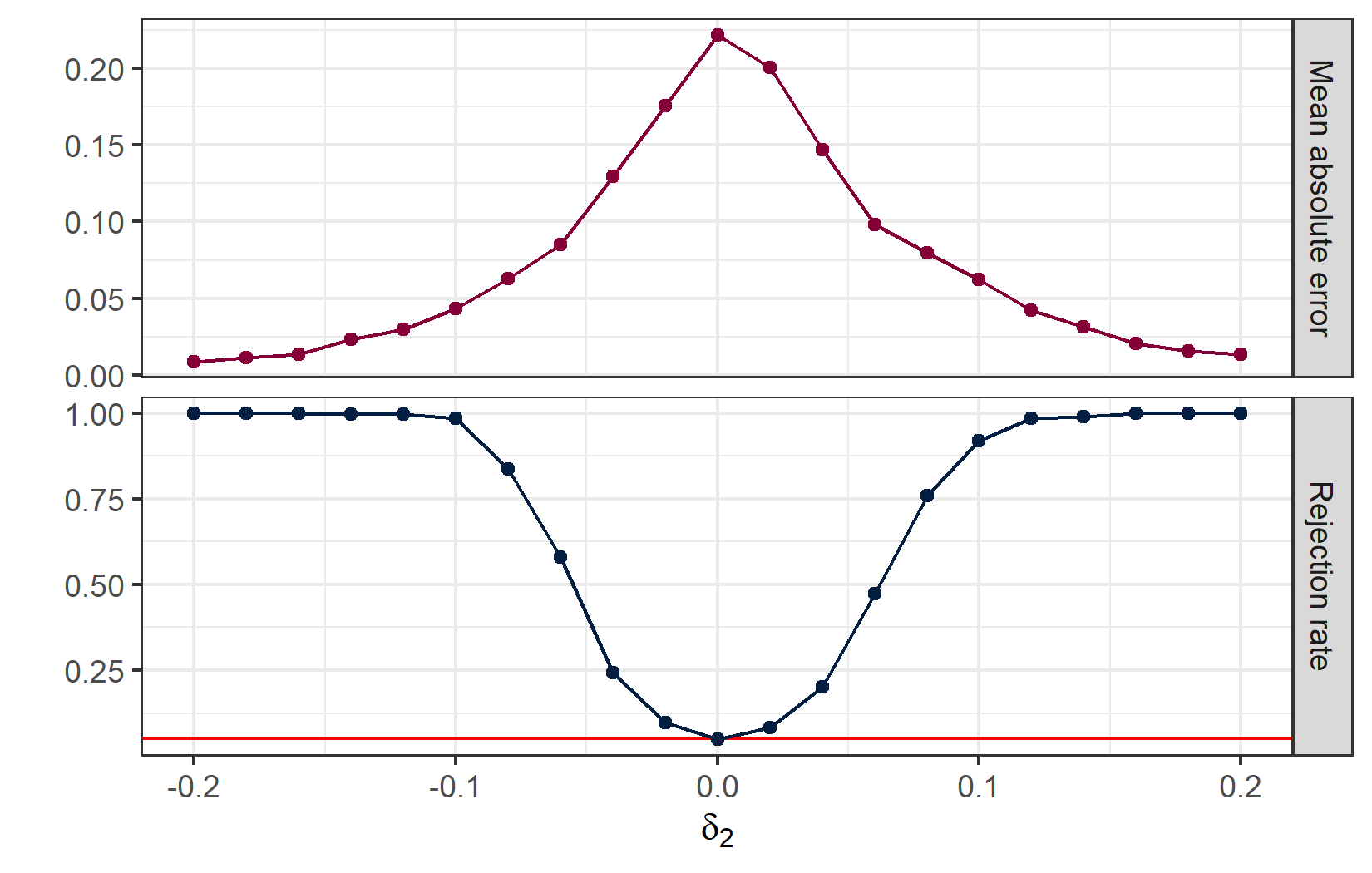}
\caption{Empirical rejection rates for the likelihood ratio test applied to 500 PA$(0.6; 0, \delta_2)$ networks with $50{,}000$ edges and $\gamma = 0.1$ along with average absolute error for changepoint location. Error is only computed for those tests that reject the null hypothesis. Red line indicates $\alpha = 0.05$.}
\label{fig:HA}
\end{figure}

\begin{table}
\centering
\begin{tabular}{ c c c c c c}
\hline \hline
$\delta_2$ & $0.1$ & $0.2$ & $0.3$ & $0.4$ & $0.5$ \\ 
\hline
MAE($\hat{t}_n$) & $0.0281$ & $0.0070$ & $0.0030$ & $0.0018$ & $0.0013$  \\
MAE($\hat{T}_n$) & $0.2574$ & $0.1518$ & $0.0930$ & $0.0710$ & $0.0560$ \\
\hline 
\end{tabular}
\caption{Empirical mean absolute error of $\hat{t}_n$ and $\hat{T}_n$ for 500 PA$(0.6; 0, \delta_2)$ networks with $100{,}000$ edges  and $\gamma = 0.1$.}
\label{tab:HA2}
\end{table}

\begin{table}
\centering
\begin{tabular}{ c c c c c c}
\hline \hline
$b$ & $0.5$ & $0.6$ & $0.7$ & $0.8$ & $0.9$ \\ 
\hline
MAE($\hat{t}_n$) & $0.0004$ & $0.0005$  & $0.0009$ & $0.0019$  & $0.0079$  \\
MAE($\hat{T}_n$) & $0.0343$ & $0.0440$ & $0.0607$ & $0.0969$ & $0.2320$ \\
\hline 
\end{tabular}
\caption{Empirical mean absolute error for $\hat{t}_n$ and $\hat{T}_n$ for 500 networks initialized by a PA$(1)$ process which transitions to a sublinear PA process with attachment function $(\text{Degree} + 1)^b$ at $t^\star = 0.6$. Networks are of $100{,}000$ edges, and we choose $\gamma = 0.1$.}
\label{tab:HAsub}
\end{table}

\section{Multiple changepoint detection}\label{sec:ScoreTest}

So far our study only focuses on the case with a single changepoint, but for pragmatic purposes, it is important to consider the generalization to circumstances with multiple changepoints.
We start with discussions on some difficulties that the methodology presented in Section \ref{sec:ChangepointEstimation} may face when applied to the multiple changepoint setting and present an alternative strategy to detect multiple changepoints. The consideration of multiple changepoints can not only accentuate difficulties faced in the single changepoint case, but also introduce new complications. Akin to the single changepoint case, statistical multiple changepoint detection is often presented under the hypothesis testing framework, a convention that we will follow \cite{chakraborty2021high, niu2016multiple}. 

Fix $\gamma \in (0, 1/2)$ and consider the hypothesis test
\begin{align}
\label{eq:multiplehypotheses}
\begin{split}
H_0: \ &\text{The graph sequence  } \{ G(k) \}_{k = 2}^n \text{  evolves according to } \text{PA}(\delta). \\
H_A: \ &\text{There exists a partition } \gamma < t^\star_1 <  \dots < t^\star_{\tau} < 1 - \gamma \text{ such that } \{ G(k) \}_{k = 2}^n \\
&\text{ evolves according to } \text{PA}(t^\star_1, \dots, t^\star_{\tau}; \delta_1, \dots, \delta_{\tau + 1})
\end{split}
\end{align}
Here $\text{PA}(t^\star_1, \dots, t^\star_{\tau}; \delta_1, \dots, \delta_{\tau + 1})$ is the natural extension of the model presented in Section \ref{sec:PAChangepoint} but now jumps occur at edges $\floor{nt^\star_1}, \dots, \floor{nt^\star_{\tau + 1}}$ with the initializer process being PA$(\delta_1)$. Note that $H_0$ in \eqref{eq:multiplehypotheses} matches that of \eqref{eq:hypotheses}, implying that if the null hypothesis is true, the same exact methodology from Section \ref{sec:ChangepointEstimation} can be applied. Jointly estimating the changepoints under $H_A$ is a difficult task, and hence the task of finding multiple changepoints is often reduced to sequentially uncovering changepoints one at a time. That is, the multiple changepoint detection problem is reduced to a sequence of single changepoint detection problems. 

\subsection{Screening and Ranking (SaRa)}\label{sec:Sara}

One popular way decomposing the multiple changepoint problem to a sequence of single changepoint problems is via local estimation of the multiple changepoints. When applying likelihood ratio methodology, it is pertinent to ensure that the model is correctly specified. Naively applying the sequence of statistics \eqref{LR} to a graph sequence with multiple changepoints will result in portions of the edgelist that are specified as a single PA process, though they are truly governed by multiple PA processes. Under the alternative, this leads to suboptimal estimation of model parameters, and thus suboptimal statistics and hypothesis tests. Instead, to ensure test regions contain at most one changepoint, \cite{hao2013multiple} and \cite{niu2012screening} proposed a screening-and-ranking (SaRa) algorithm to detect multiple changepoints. In effect, the SaRa algorithm decomposes the hypothesis test \eqref{eq:multiplehypotheses} into a sequence of hypotheses
\begin{align}
\label{eq:multiplehypothesesSaRa}
\begin{split}
H_0(k): \ &\text{The edges } E(k + h)\setminus E(k - h) \text{ are added via offset parameter } \delta. \\
H_A(k): \ &\text{The edges } E(k + h)\setminus E(k) \text{ and }   E(k)\setminus E(k - h)\text{ are added via} \\
&\text{offset parameters } \delta_1 \neq \delta_2.
\end{split}
\end{align}
for $k = n_0 + h, \dots, n - h$. Here, $h > 0$ is chosen such that no two changepoints are within $h$ edges of each other. In order for asymptotics developed in previous sections to kick in, $h$ must grow with $n$. Henceforth, we assume $h_n = \floor{ns}$ for some $s \in (0, 1)$. In the continuous time setting, this assume the changepoints specified under $H_A$ in \eqref{eq:multiplehypotheses} are approximately separated by $s$ for large $n$. With this window-based assumption, individual hypothesis tests can be performed over select windows that contain only a single changepoint, assuming a proper choice of $h_n$. 

We slightly augment the methods of \cite{hao2013multiple} and \cite{niu2012screening} for our purposes. Using likelihood ratio methods from Section \ref{sec:ChangepointEstimation}, the natural statistic to test \eqref{eq:multiplehypothesesSaRa} is
\begin{equation}
\label{eq:LRh}
\frac{L_{(k - h_n, k + h_n]}(\hat{\delta}_{(k - h_n, k + h_n]})}{L_{(k - h_n, k]}(\hat{\delta}_{(k - h_n, k]}) L_{(k, k + h_n]}(\hat{\delta}_{(k, k + h_n]})} \qquad \text{for } k = h_n + 1,\dots, n - h_n - 1.
\end{equation}
Theoretically, however, this statistic is difficult to work with. Namely, if we convert to continuous time by letting $k = \floor{nt}$, $\hat{\delta}_{(nt, n(t + s)]}\approx \hat{\delta}_{(k, k + h_n]}$ and $\hat{\delta}_{(n(t - s), n(t + s)]}\approx \hat{\delta}_{(k - h_n, k + h_n]}$, it is difficult to establish the joint convergence of these (scaled and centered) processes in $D[0, 1 - \gamma] \times D[\gamma, 1- \gamma]$. In the proof of Theorem \ref{thm:nullhyp}, we were able to get away without establishing joint convergence since $\hat{\delta}_n = \delta_{(0, n]}$ is a realization of $\delta_{(0, nt]}$ at $t = 1$. Then establishing the convergence of the scaled and centered $\delta_{(0, nt]}$ in the Skorohod space was sufficient, since the continuous mapping theorem allows us to prove statements like \eqref{eq:MLEfunctional}. The joint convergence of the processes $\hat{\delta}_{(nt, n(t + s)]}$ (or equivalently $\hat{\delta}_{(n(t - s), nt]}$) and $\hat{\delta}_{(n(t-s), n(t + s)]}$ in \eqref{eq:LRh}  cannot be derived from the convergence of a single process or transformations of such processes, thus theoretical analysis becomes less tractable. We refer to Chapter 13 of \cite{whitt2002stochastic} for more on deriving stochastic process limits from existing ones.

Instead, we propose the statistic
\begin{equation}
\label{eq:LRh2}
\Lambda_k (h_n) := \frac{L_{(k - h_n, k]}(\hat{\delta}_{(k, k + h_n]}) L_{(k, k + h_n]}(\hat{\delta}_{(k - h_n, k]})}{L_{(k - h_n, k]}(\hat{\delta}_{(k - h_n, k]}) L_{(k, k + h_n]}(\hat{\delta}_{(k, k + h_n]})},
\end{equation}
for $k = h_n + 1,\dots, n - h_n - 1$.
In the denominator \eqref{eq:LRh2} maximizes the likelihood under the alternative hypotheses in \eqref{eq:multiplehypothesesSaRa}. The numerator, on the other hand, is an unsophisticated representation of the null hypothesis. Under the null, $\hat{\delta}_{(k, k + h_n]}$ and $\hat{\delta}_{(k - h_n, k]}$ should in a sense be interchangeable since they are both estimators of $\delta$. Thus, switching the maximizers of $L_{(k - h_n, k]}(\lambda)$ and $L_{(k, k + h_n]}(\lambda)$ should result in a likelihood ratio that is approximately $1$ under the null hypothesis, but far from optimal under the alternative. Note that $\hat{\delta}_{(k - h_n, k]}$ and $\hat{\delta}_{(k, k + h_n]}$ are time shifts of the same process, making continuous mapping methodology suitable for the analysis of \eqref{eq:LRh2}.

The following corollary gives the limiting process for $-2 \log \Lambda_{\floor{nt}} (\floor{ns})$. The proof is very similar to that of Theorem \ref{thm:nullhyp}, and hence we omit it for brevity.

\begin{corollary}
\label{thm:nullhypSaRa}
Fix $s \in (0, 1/2)$. Under $H_0$ in \eqref{eq:multiplehypotheses}
\begin{equation}
-2 \log \Lambda_{\floor{nt}}(\floor{ns}) \Rightarrow \frac{2}{s} \left(W(t + s) + W(t - s) - 2 W(t) \right)^2
\end{equation}
in $D[s, 1 - s]$, where $\Rightarrow$ denotes the weak convergence and $W(\cdot)$ is a Wiener process.
\end{corollary}
Corollary~\ref{thm:nullhypSaRa} provides the building block for the application of 
the SaRa algorithm to our PA network setup here, which we now explain.

The SaRa algorithm proceeds by collecting the locations which locally maximize $-2 \log \Lambda_k (h_n)$. Along the lines of \cite{hao2013multiple}, we call $\ell$ a $h$-local maximizer if 
$$
-2 \log \Lambda_\ell (h_n) \geq -2 \log \Lambda_k (h_n),
\quad \text{for all }
k \in \{ \ell - h, \dots, \ell + h \}.
$$ 
For simpilcity, we choose $h = h_n$ as a default, though more informed choices can be made. Naturally, these maximums locally maximize the discrepancy characterized by the likelihood ratio and provide the most evidence against the null hypothesis. Once the scanning step is complete and we have identified the $h_n$-local maximizers, we compare the collection of $h_n$-local maximums to the $(1 - \alpha)$-th quantile of the $s$-local maximimums of  
$$X(t) := \frac{2}{s} \left(W(t + s) + W(t - s) - 2 W(t) \right)^2,$$
 the limiting process of $-2 \log \Lambda_{\floor{nt}} (\floor{ns})$. A $s$-local maximizer is a continuous time analog the $h_n$-local maximizer; $t$ is a $s$-local maximizer of $X(t)$ if $X(t) \geq X(u)$ for all $u \in [t - s, t + s]$. The quantiles of the $s$-local maximums of $X(t)$ are difficult to obtain analytically, but since $X(t)$ is a transformation of Wiener processes, they can be easily simulated. 

In comparison to \eqref{LR}, the statistic \eqref{eq:LRh2} has computational advantages. As discussed previously, the maximum likelihood estimates in \eqref{LR} and \eqref{eq:LRh2} in  do not have a closed-form and are solved by setting \eqref{score} equal to zero and employing Newton's Method. In particular, for \eqref{LR}, one must compute the likelihood ratio at every single edge $m$, which further  involves computing the maximum likelihood estimates $\hat{\delta}_{(0, m]}$ and $\hat{\delta}_{(m, n]}$. Even with clever programming and techniques such as warm starts, this is not easily computable for large networks. Although optimization still needs to be performed at every time point $m$ for \eqref{eq:LRh2}, the window $h_n$ also serves as a computational backstop for the procedure since optimization occurs over $2 h_n$ edges rather than employing the entire edgelist. Though faster methods are still to be desired if one were to analyze large networks with other changepoint segmentation methods such as binary segmentation.
We also point out that 
these analyses give rise to multiple hypothesis tests, so one may consider control of the false discovery rate or family wise error rate as in \cite{hao2013multiple}. We will leave this point for future research.  

\subsection{Score test}

Another popular way of sequentially detecting multiple changepoints is via binary segmentation \cite{niu2016multiple}. Although there are many variants of binary segmentation, including those of \cite{fryzlewicz2014wild} and \cite{olshen2004circular}, we apply the routine version as presented in Algorithm \ref{alg:BS}. In the algorithm, we let $f(\cdot)$ be a function that, from an edgelist, returns TRUE if a changepoint is detected, and otherwise returns FALSE. Further, let $g(\cdot)$ be a function that returns an estimated changepoint from an edgelist. 

\begin{algorithm} 
\label{alg:BS}
\caption{BinarySegmentation($E(n)$, $f(\cdot)$, $g(\cdot)$)}
\begin{algorithmic}
\Require edgelist $E(n)$, $f(\cdot)$, $g(\cdot)$\;
\Ensure set of estimated changepoint locations  $\hat{\mathcal{T}}$\;
\If{$f(E(n))$}
	\State $m \gets g(E(n))$\;
    \State $\hat{\mathcal{T}} \gets \hat{\mathcal{T}} \cup m$\;
    \State BinarySegmentation($E(m)$,  $f(\cdot)$, $g(\cdot)$)\;
    \State BinarySegmentation($E(n) \setminus E(m)$,  $f(\cdot)$, $g(\cdot)$)\;
\EndIf
\end{algorithmic}
\end{algorithm}

Binary segmentation is a particularly convenient technique as it conveniently extends the single changepoint detection procedures to the multiple changepoint setting. Though, for computationally intensive procedures, their repeated application over multiple segments renders them infeasible for practical applications. This is especially true for dynamic networks, as the number of edges in a typical network can exceed the millions or even billions. As mentioned in Section \ref{sec:Sara}, the likelihood ratio procedures in Section \ref{sec:ChangepointEstimation} can unfortunately become computationally onerous, which calls for the need of a quicker procedure to detect changepoints.

In order to assuage the computational bottleneck induced by the likelihood ratio methodology, we introduce a score-based statistic that drastically reduces computational load in detecting changepoints. To introduce the statistic assume for a moment that we are indeed under $H_0$ in \eqref{eq:multiplehypotheses}. Under $H_0$, the best estimate of $\delta$ is $\hat{\delta}_n = \hat{\delta}_{(0, n]}$, the MLE based on the entire network. By definition, $u_{(0, n]}(\hat{\delta}_n) = 0$. However, if the whole data stream is governed by PA$(\delta)$, then we might expect that $\hat{\delta}_n$ approximately solves $u_{(0, m]}(\lambda) = 0$ for any $m < n$. On the other hand, if there is a changepoint at $m$, $u_{(0, m]}(\hat{\delta}_n)$ should be far from $0$ since  $\hat{\delta}_n$ is computed under a misspecified model. This leads us to propose the following statistic for changepoint detection:
\begin{equation}
\label{eq:scorestat}
S_m := -u^2_{(0, m]}(\hat{\delta}_n) \left( \frac{1}{u'_{(0, m]}(\hat{\delta}_n)} + \frac{1}{u'_{(m, n]}(\hat{\delta}_n)} \right)\qquad \text{for } m = \floor{n\gamma},\dots, \floor{n(1 - \gamma)}.
\end{equation}
Here, we  need to ensure $\floor{n\gamma} \leq m \leq \floor{n(1 - \gamma)}$ for $\gamma \in (0, 1)$ in order to guarantee that there is sufficient data so that the score function and observed information are a good representatives of the PA process up to that point. The hessian in \eqref{eq:scorestat} is used as a scaling factor to eliminate unknown parameters in the asymptotic distribution of $S_m$. In theory, one could employ $u'_{(0, n]}(\hat{\delta}_n)$ instead of $u'_{(0, nt]}(\hat{\delta}_n)$ for a better estimate of the asymptotic covariance of $\hat{\delta}_n$ under the null hypothesis, though we employ the latter to match the asymptotic distribution in Theorem \ref{thm:nullhyp}. Intuitively, the likelihood ratio \eqref{LR} is likely the most powerful, and hence it may be beneficial to imitate its behavior under the null hypothesis. Theorem \ref{thm:nullhyp2} gives the asymptotic distribution of the proposed statistic under the null hypothesis, with proofs deferred to Appendix~\ref{subsec:pf_null2}. This gives rise to a hypothesis test of changepoint existence where the null hypothesis is rejected when $\sup_{t \in [\gamma, 1]} S_{nt}$ exceeds the $(1 - \alpha)$-th quantile of a $\sup_{t \in [\gamma, 1]} B^2(t)/t(1-t)$ distribution.

\begin{theorem}
\label{thm:nullhyp2}
Fix $\gamma \in (0, 1/2)$. Then under $H_0$ in \eqref{eq:multiplehypotheses}
\begin{equation}
\sup_{t \in [\gamma, 1 - \gamma]} S_{nt} \Rightarrow \sup_{t \in [\gamma, 1 - \gamma]} \frac{B^2(t)}{t(1 - t)},
\end{equation}
in $\mathbb{R}$.
\end{theorem}

By using \eqref{eq:scorestat}, the computational load is alleviated by requiring only one computation of the MLE per segment. If \eqref{LR} were applied, one would need to compute on the order of $n$ many MLE's for the first segment. Further, since the score function is additive, it requires only $O(n)$ operations to compute $S_{m}$ for $m = \floor{n \gamma},\dots, n$. Although we are unable to present a consistency result for the score-based method, the computational benefits are obvious.

\subsection{Multiple changepoint simulation study}\label{subsec:multiple_sim}

In this section we compare the empirical performance of the window and score methods to simulated, multiple changepoint data.  As in \cite{niu2016multiple}, it may be desirable to control the family-wise error rate (FWER) or false discovery rate (FDR) via a multiple testing procedure, though we present unadulterated results as the choice of procedure can muddy comparisons. First, we evaluate the false positive rates for the window and score methods when there is no changepoint. Since the window method finds multiple local maxima within a data sequence, we define the positive rate to be the proportion of $h_n$-local maximums that reject the null hypothesis out of total identified $h_n$-local maximums. On the other hand, the positive rate for the score test is the number of times the maximum $S_m$ exceeds the $(1 -\alpha)$-th quantile of the null distribution in Theorem \ref{thm:nullhyp2} out of the number of total segments tested. Under the null, the positive rate for the score test will be close to the FWER due to the sequential nature of binary segmentation. For the window method, however, the FWER will be inflated compared to the positive rate under the null due to multiple hypothesis tests being performed at once (one for each $h_n$-local maxima). 

To evaluate control on the false positive rate, we apply the window and score methods to 500 simulated PA$(\delta)$ networks of size $100{,}000$ where $\delta \in \{-0.5, 0, 1, 2 \}$. We let $h_n = 10{,}000$ for the window statistic and $\gamma = 0.1$ for the score statistic. Here, $h_n$ and $\gamma$ are chosen so that the region $[0.1, 0.9]$ is searched for a changepoint in both cases. Further, these settings are chosen to ensure that there are enough edges to accurately estimate parameters of the PA processes under each regime. The empirical false positive rate for each $\delta$ value is reported in Table \ref{tab:H0mult}. The false positive rates behave as expected, concentrating around $\alpha = 0.05$. The study indicates that even in the multiple changepoint setting, the hypothesis testing perspective still provides control on detecting the existence of a changepoint.

\begin{table}
\centering
\begin{tabular}{ c c c c c}
\hline \hline
$\delta$ & $-0.5$ & $0$ & $1$ & $2$ \\ 
\hline
Window & $0.0457$ & $0.0473$ & $0.0556$ & $0.0491$  \\ 
Score & $0.0499$ & $0.0558$ & $0.0519$ & $0.0519$  \\  
\hline 
\end{tabular}
\caption{False positive rates for the window and score methods applied to 500 PA$(\delta)$ networks with $50{,}000$ edges and $\gamma = 0.1$ and $h_n = 10{,}000$.}
\label{tab:H0mult}
\end{table}

Next, we assess the ability for the window and score methods to detect multiple changepoints and estimate their locations consistently. To do so, we apply the methods to 500 simulated PA$(0.2, 0.5; 1, 1.5, \delta_3)$ networks of size $100{,}000$. That is, the network-generating process switches from PA$(1)$ to PA$(1.5)$ at $t^\star_1 = 0.2$, and then switches from PA$(1.5)$ to PA$(\delta_3)$ at $t^\star_2 = 0.5$. We allow $\delta_3$ to vary from $1$ to $1.4$ by increments of $0.1$. If $\delta_3 = 1$, we are under an ``epidemic alternative"; the network generating process eventually returns to the state it was initialized under. As in the previous simulation, we let $h_n = 10{,}000$ and $\gamma = 0.1$. The estimated number of changepoints and average Rand index across the 500 simulations are reported in Table \ref{tab:HAmult}. For $\delta_3 = 1$, which is the largest change between the PA processes at $t^\star = 0.5$, the window method exhibits the benefits of local estimation of the changepoint. Since changepoint estimation relies on local information, the initializing PA$(1)$ process does not influence the estimation. The score method performs worse in this situation since the null hypothesis estimate $\hat{\delta}_n$ reflects that the network is generating by a PA$(1)$ process for $70\%$ of the network evolution. However, as soon as $\delta_3$ is slightly larger than $1$, the score method outperforms the window-based method. Once $\delta_3$ becomes closer to $1.5$ and the second and third PA processes become more similar, both methods exhibit poorer performance. Though, the score method more consistently detects that there are indeed $2$ changepoints. Generally, the window method more accurately detects changepoints under an epidemic-like alternative, but the speed and accuracy of the score method under other settings make it a desirable choice for multiple changepoint detection.

\begin{table}
\centering
\begin{tabular}{ c c @{\hskip 7mm} c c c c c c}
\hline \hline
Method & $\delta_3$ & \multicolumn{5}{c}{\# of changepoints} & Rand index \\
\hline 
& & $0$ & $1$ & $2$ & $3$ & $4$ & \\[2mm]
\multirow{5}{*}{Window} & $1.0$ & $0$ & $12$ & $458$ & $29$ & $1$ & $0.963$ \\ 
 & $1.1$ & $0$ & $56$ & $414$ & $29$ & $1$ & $0.934$ \\
 & $1.2$ & $3$ & $191$ & $282$ & $24$ & $0$ & $0.851$ \\
 & $1.3$ & $1$ & $352$ & $141$ & $6$ & $0$ & $0.757$ \\
 & $1.4$ & $5$ & $425$ & $67$ & $3$ & $0$ & $0.713$ \\[.25 em]
 \hline 
\multirow{5}{*}{Score} & $1.0$ & $0$ & $128$ & $370$ & $3$ & $0$ & $0.934$ \\  
 & $1.1$ & $0$  & $36$  & $451$ & $13$ & $0$ & $0.948$ \\
 & $1.2$ & $2$  & $6$ & $479$ & $13$ & $0$ & $0.940$ \\
 & $1.3$ & $0$ & $127$ & $363$ & $10$ & $0$ & $0.858$ \\
 & $1.4$ & $0$ & $390$ & $108$ & $2$ & $0$ & $0.725$\\
\hline 
\end{tabular}
\caption{Number of changepoints found and average rand index for the window and score methods applied to 500 PA$(0.2, 0.5; 1, 1.5, \delta_3)$ networks with $100{,}000$ edges and $\gamma = 0.1$ and $h_n = 10{,}000$.}
\label{tab:HAmult}
\end{table}

\section{Data Example}\label{sec:data}

We now demonstrate the applicability of our changepoint detection methods to real world networks by applying our methods to the Twitter Higgs network, first analyzed in \cite{de2013anatomy} (available on \cite{snapnets}). The Higgs boson particle, whose existence was theorized in 1964 by Peter Higgs and others, bolstered a popular theory in particle physics. On July 4th, 2012 at 8 AM GMT, scientists at CERN sent ripples through the physics community with the announcement of the discovery of a particle with properties agreeing with that of the Higgs boson particle. Before the official announcement, however, rumors of the discovery spread through Twitter. The Twitter Higgs network tracks the activity of a subset of users discussing this discovery from July 1st to July 7th, 2012. Here, we limit our analysis to just the retweets of these users so that a directed edge $(v, w)$ indicates that node $v$ retweeted node $w$. This network is temporal, and each edge is observed in the order of its creation with a timestamp. Intuitively, one would expect that with the announcement of Higgs boson discovery, the dynamics and short term attractiveness of users retweets would shift dramatically. 

Unfortunately, a retweet network is not a tree; users can retweet multiple times during the network evolution. Hence, in order to better conform with model assumptions, we remove users with out-degree larger than $1$. Although this seems like a drastic adjustment, 67.6\% of users retweet only once throughout the network evolution. Hence, we are still able to analyze the behavior of a typical user in the network with this data cleaning step. This resulting network contains 198,437 users. 

We apply both the score and window tests to the network to detect multiple changepoints. For the score test we allow $\gamma = 0.1$, and we set $h_n = \floor{0.05 n} = 9921$ for the window method. The score method detects 5 changepoints, while the window method detects 6. In order to correct for multiple testing, we employ a Holm adjustment to the $p$-values for both procedures \citep{holm1979simple}. In this case, the Holm correction has no effect on the number of estimated changepoints, indicating the significance for the changepoints detected. 

We then utilize maximum likelihood to fit preferential attachment models over the constant segments of the network defined by the changepoints. The maximum likelihood estimates and Fisher information-based confidence intervals (cf. \cite{gao2017asymptotic, wan2017fitting}) are given in Figures \ref{fig:score} and \ref{fig:window} for the score and window methods, respectively. 
Throughout, the estimated changepoints are colored in red, with the official announcement time colored in blue. 
Additionally, Figures \ref{fig:score} and \ref{fig:window} display the moving proportion of nodes with degrees larger than $100$ within a centered time window of length $200$. Note that the choice of the threshold ($100$, the empirical $70$-percentile of the degree distribution) can be changed to any other upper quantiles of the degree distribution. Essentially, if the proportion of selected nodes with large degrees ($>100$) increases within a time period, then we may expect the degree distribution to have heavier tails, and the estimated offset parameter $\delta$ over that region should decrease. As a result, the attachment mechanism has become more preferential and thus high degree nodes are more attractive.

\begin{figure}
\includegraphics[width=\textwidth]{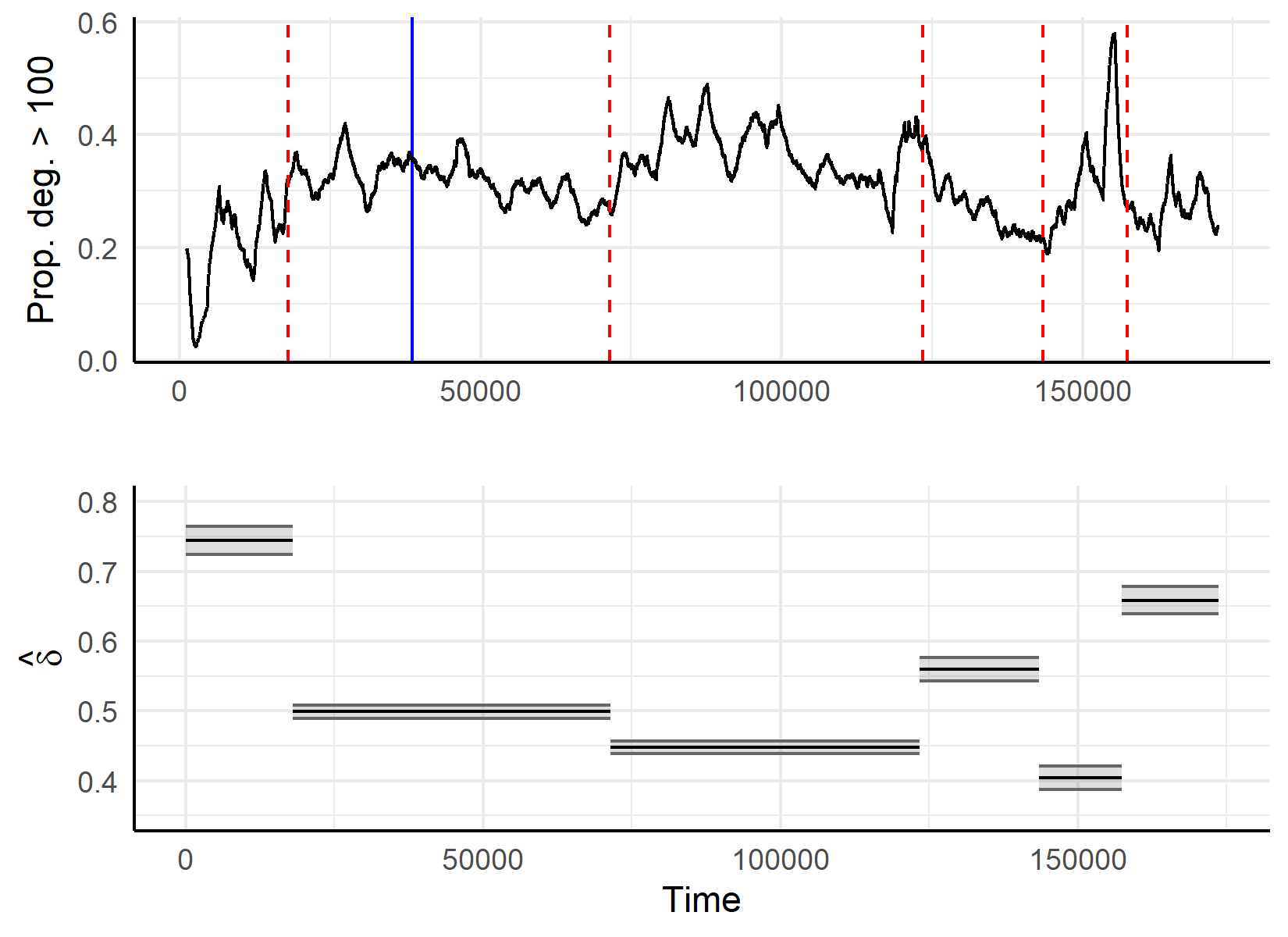}
\caption{Time of Higgs boson announcement (in blue) along with estimated changepoints (in red)  for the score method. The first panel displays a sliding window proportion of attached nodes with degree larger than $100$ over a period of length $200$. The second panel displays offset parameter estimates and confidence intervals over the constant regions computed via MLE.}
\label{fig:score}
\end{figure}

\begin{figure}
\includegraphics[width=\textwidth]{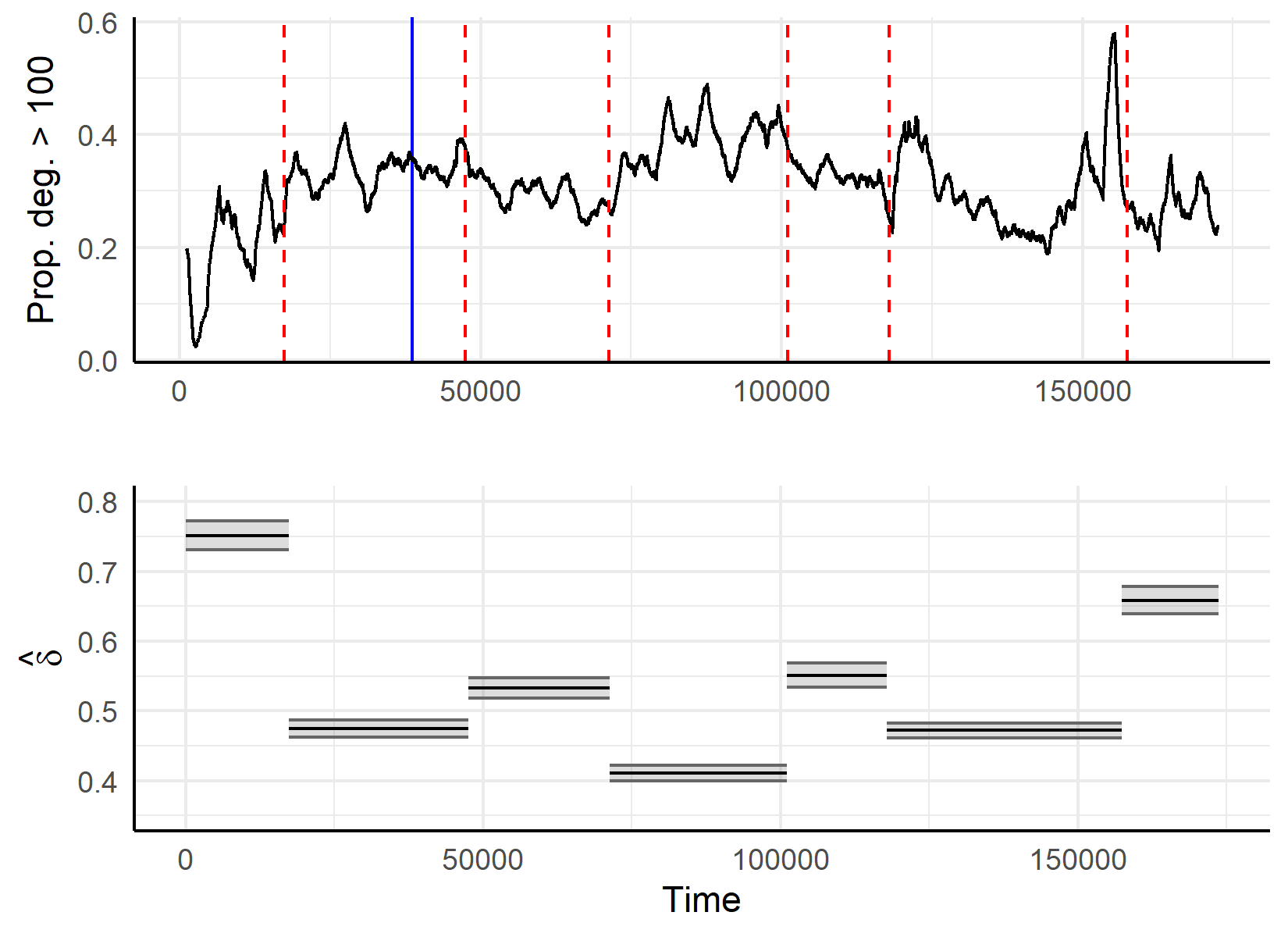}
\caption{Time of Higgs boson announcement (in blue) along with estimated changepoints (in red)  for the window method. The first panel displays a sliding window proportion of attached nodes with degree larger than $100$ over a period of length $200$. The second panel displays offset parameter estimates and confidence intervals over the constant regions computed via MLE.}
\label{fig:window}
\end{figure}

Both methods capture the varying dynamics in users retweet behavior. 
The nearest changepoint to the official announcement detected by the score method was approximately two hours before the announcement at 05:49:55 GMT, while
the window method produced the nearest estimated changepoint at 09:08:20 GMT on 07/04/2012. We further remark that in Figures~\ref{fig:score} and \ref{fig:window}, the $x$-axis corresponds to the number of steps in the network evolution, not the real time scale. Due to the large number of retweets produced around the announcement of the news, the first two changepoints detected by the two methods are in fact all not very far away from the official announcement in real time.

Across both methods, the proportion of attached nodes with degree greater than 100 remains relatively constant between the estimated changepoints, indicating that the attachment process has stabilized over these regions. Both methods indicate that the offset parameter $\delta$ decreases near the official announcement of the discovery, while eventually increasing towards the end of the observed network evolution. Such trend suggests that around the announcement the attachment process becomes more preferential, so that large degree nodes are more attractive over that time period. This is reasonable since initial tweets surrounding a news breaking event are likely to be popular and quickly gain attention. In later stages of the news cycle, however, the attachment process becomes more uniform, indicating that attention has diffused throughout the network after a short period of time.

Although the window method is more sensitive, some estimated changepoint locations are shared across the two methods. Figure \ref{fig:LRT_plot} reports the test statistics for the first segmentation step of the score method, as well as the likelihood ratio statistics for the window method. The lack of spikes in the score statistic between July 5th and 6th conveys the necessity of binary segmentation; without additional segmentation, the two additional changepoints in that time period would likely not have been detected. Note that for the likelihood ratio statistics, there is a later spike that remains undetected, but is classified as a changepoint via the score statistic. This spike was not classified as an $h_n$-local maximizer due to its proximity to another maximizer, and hence is undetected by the window method. The window $h_n$ can be further reduced to estimate additional changepoints in the network evolution, though we choose not to do so in order to retain accurate estimation of the offset parameters. 

Overall, our numerical analyses show that the preferential attachment model with changepoints, combined with our likelihood-based methodology, can be used to successfully track the dynamics of an evolving network.

\begin{figure}
\includegraphics[width=\textwidth]{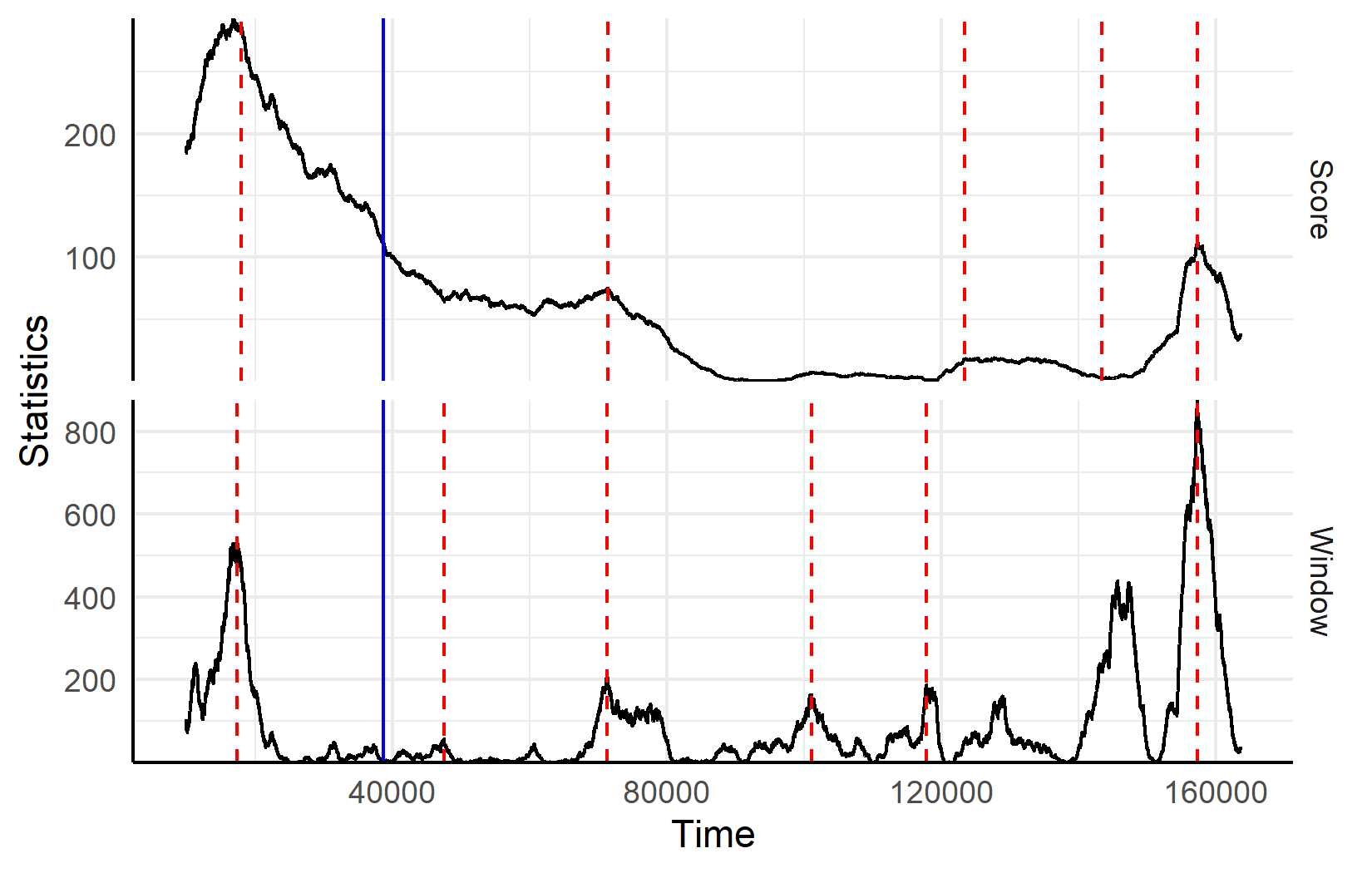}
\caption{The statistics $S_k$ and $\Lambda_{k}(h_n)$ plotted over the network evolution along with the estimated changepoints for each method (in red). The announcement of the Higgs boson discovery is marked in blue.}
\label{fig:LRT_plot}
\end{figure}

\section{Conclusion}\label{sec:conclusion}

In this paper, we introduce likelihood-based methods for changepoint detection in preferential attachment models. For single changepoint detection, we produced a theoretically justified likelihood ratio test that empirically performs better than a non-parametric estimator under correctly and miss-specified models. Further, we extend the likelihood based methodology to the multiple changepoint setting where both window-based and fast binary segmentation methods can be applied to large networks. The presented methods offer solutions to the statistical detection of changepoints via a hypothesis testing perspective not previously put forward in the network literature. When applied to a Twitter retweet network, the multiple changepoint methods offer a statistically sound analysis of the varying dynamics in the presence of shocks to the network evolution.

With this work, there are multiple avenues of research yet to be explored. First, one can consider the extension of these methods to more realistic network models such as the one in \cite{wan2017fitting}. We believe the methods presented here can be extended to the case where users in a network can make multiple connections. Further, one can consider the development of likelihood-based methodology for changepoint detection in more general preferential attachment models such as sublinear preferential attachment (as in \cite{banerjee2023fluctuation, gao2017consistent}). Finally, one can consider more interesting definitions of a changepoint in the network setting. Here, we considered a classical definition of changepoint where the distribution of the entire network changes, but it is also possible to assume changes in a subset of the network or some other types of changes in the network dynamics. 

\bibliographystyle{plain}
\bibliography{PA_Changepoint.bib}

\begin{thebibliography}{10}

\bibitem{athreya2008growth}
Krishna~B Athreya, Arka~P Ghosh, and Sunder Sethuraman.
\newblock Growth of preferential attachment random graphs via continuous-time
  branching processes.
\newblock {\em Proceedings Mathematical Sciences}, 118:473--494, 2008.

\bibitem{banerjee2023fluctuation}
Sayan Banerjee, Shankar Bhamidi, and Iain Carmichael.
\newblock Fluctuation bounds for continuous time branching processes and
  evolution of growing trees with a change point.
\newblock {\em The Annals of Applied Probability}, 33(4):2919--2980, 2023.

\bibitem{barabasi1999emergence}
Albert-L{\'a}szl{\'o} Barab{\'a}si and R{\'e}ka Albert.
\newblock Emergence of scaling in random networks.
\newblock {\em Science}, 286(5439):509--512, 1999.

\bibitem{bhamidi2018change}
Shankar Bhamidi, Jimmy Jin, and Andrew Nobel.
\newblock Change point detection in network models: Preferential attachment and
  long range dependence.
\newblock {\em The Annals of Applied Probability}, 28(1):35--78, 2018.

\bibitem{bollobas2003directed}
B{\'e}la Bollob{\'a}s, Christian Borgs, Jennifer~T Chayes, and Oliver Riordan.
\newblock Directed scale-free graphs.
\newblock In {\em SODA}, volume~3, pages 132--139, 2003.

\bibitem{chakraborty2021high}
Shubhadeep Chakraborty and Xianyang Zhang.
\newblock High-dimensional change-point detection using generalized homogeneity
  metrics.
\newblock {\em arXiv preprint arXiv:2105.08976}, 2021.

\bibitem{cirkovic2023preferential}
Daniel Cirkovic, Tiandong Wang, and Sidney~I Resnick.
\newblock Preferential attachment with reciprocity: properties and estimation.
\newblock {\em Journal of Complex Networks}, 11(5):cnad031, 2023.

\bibitem{csorgo1997limit}
Mikl{\'o}s Cs{\"o}rg{\"o}, Mikl{\'o}s Cs{\"o}rg{\"o}, and Lajos Horv{\'a}th.
\newblock {\em Limit Theorems in Change-Point Analysis}.
\newblock John Wiley \& Sons, 1997.

\bibitem{de2013anatomy}
Manlio De~Domenico, Antonio Lima, Paul Mougel, and Mirco Musolesi.
\newblock The anatomy of a scientific rumor.
\newblock {\em Scientific Reports}, 3(1):1--9, 2013.

\bibitem{durrett1978functional}
Richard Durrett and Sidney~I Resnick.
\newblock Functional limit theorems for dependent variables.
\newblock {\em The Annals of Probability}, pages 829--846, 1978.

\bibitem{fryzlewicz2014wild}
Piotr Fryzlewicz.
\newblock Wild binary segmentation for multiple change-point detection.
\newblock {\em The Annals of Statistics}, 42(6):2243--2281, 2014.

\bibitem{gao2017asymptotic}
Fengnan Gao and Aad van~der Vaart.
\newblock On the asymptotic normality of estimating the affine preferential
  attachment network models with random initial degrees.
\newblock {\em Stochastic Processes and their Applications},
  127(11):3754--3775, 2017.

\bibitem{gao2017consistent}
Fengnan Gao, Aad van~der Vaart, Rui Castro, and Remco van~der Hofstad.
\newblock Consistent estimation in general sublinear preferential attachment
  trees.
\newblock {\em Electronic Journal of Statistics}, 11(2):3979--3999, 2017.

\bibitem{hao2013multiple}
Ning Hao, Yue~Selena Niu, and Heping Zhang.
\newblock Multiple change-point detection via a screening and ranking
  algorithm.
\newblock {\em Statistica Sinica}, 23(4):1553, 2013.

\bibitem{hariz2007optimal}
Samir~Ben Hariz, Jonathan~J. Wylie, and Qiang Zhang.
\newblock {Optimal rate of convergence for nonparametric change-point
  estimators for nonstationary sequences}.
\newblock {\em The Annals of Statistics}, 35(4):1802 -- 1826, 2007.

\bibitem{hill1975simple}
Bruce~M Hill.
\newblock A simple general approach to inference about the tail of a
  distribution.
\newblock {\em The Annals of Statistics}, pages 1163--1174, 1975.

\bibitem{holm1979simple}
Sture Holm.
\newblock A simple sequentially rejective multiple test procedure.
\newblock {\em Scandinavian Journal of Statistics}, pages 65--70, 1979.

\bibitem{konect}
J\'{e}r\^{o}me Kunegis.
\newblock {KONECT} -- {The} {Koblenz} {Network} {Collection}.
\newblock In {\em Proc. Int. Conf. on World Wide Web Companion}, pages
  1343--1350, 2013.

\bibitem{snapnets}
Jure Leskovec and Andrej Krevl.
\newblock {SNAP Datasets}: {Stanford} large network dataset collection.
\newblock {http://snap.stanford.edu/data}, June 2014.

\bibitem{niu2016multiple}
Yue~S Niu, Ning Hao, and Heping Zhang.
\newblock Multiple change-point detection: a selective overview.
\newblock {\em Statistical Science}, pages 611--623, 2016.

\bibitem{niu2012screening}
Yue~S Niu and Heping Zhang.
\newblock The screening and ranking algorithm to detect {DNA} copy number
  variations.
\newblock {\em The Annals of Applied Statistics}, 6(3):1306, 2012.

\bibitem{olshen2004circular}
Adam~B Olshen, ES~Venkatraman, Robert Lucito, and Michael Wigler.
\newblock Circular binary segmentation for the analysis of array-based {DNA}
  copy number data.
\newblock {\em Biostatistics}, 5(4):557--572, 2004.

\bibitem{rudas2007random}
Anna Rudas, B{\'a}lint T{\'o}th, and Benedek Valk{\'o}.
\newblock Random trees and general branching processes.
\newblock {\em Random Structures \& Algorithms}, 31(2):186--202, 2007.

\bibitem{van2017random}
Remco Van Der~Hofstad.
\newblock {\em Random Graphs and Complex Networks}, volume~43.
\newblock Cambridge University Press, 2017.

\bibitem{wan2017fitting}
Phyllis Wan, Tiandong Wang, Richard~A Davis, and Sidney~I Resnick.
\newblock Fitting the linear preferential attachment model.
\newblock {\em Electronic Journal of Statistics}, 11(2):3738--3780, 2017.

\bibitem{wang2023poisson}
Tiandong Wang and Sidney Resnick.
\newblock Poisson edge growth and preferential attachment networks.
\newblock {\em Methodology and Computing in Applied Probability}, 25(1):8,
  2023.

\bibitem{wang2019consistency}
Tiandong Wang and Sidney~I Resnick.
\newblock Consistency of hill estimators in a linear preferential attachment
  model.
\newblock {\em Extremes}, 22(1):1--28, 2019.

\bibitem{whitt2002stochastic}
Ward Whitt.
\newblock {\em Stochastic-Process Limits: An Introduction to Stochastic-Process
  Limits and Their Application to Queues}, volume 500.
\newblock Springer, 2002.

\end{thebibliography}

\appendix
\section{Proofs}\label{sec:pf}

\subsection{Proof of Theorem \ref{thm:nullhyp}}\label{subsec:pf_thm1}
The proof of Theorem \ref{thm:nullhyp} requires the weak convergence of 
$\sqrt{n}(\hat{\delta}_{(ns, nt]} - \delta)$, which is given in Lemma \ref{lem:MLE}. The weak convergence of $\sqrt{n}(\hat{\delta}_{(ns, nt]} - \delta)$ further relies on the weak convergence of the score function evaluated at the truth, $n^{-1}u_{(0, nt]}(\delta)$, to a Weiner process. To prove this, we rely on the fact that the summands of $n^{-1}u_{(0, nt]}(\delta)$ are a martingale difference array, and hence we can apply functional martingale central limit theorems (FMCLT). We now present Lemma \ref{lem:MLE}.

\begin{lemma}\label{lem:MLE}
Fix $s \in [0, 1]$ and $\tau \in (0, 1 - s)$. Assume that $\{ G(k) \}_{k = 2}^n$ evolves according to the PA$(\delta)$ rule.
Then as $n \rightarrow \infty$
\[ 
(t - s)I(\delta; \delta) \cdot \sqrt{n}(\hat{\delta}_{(ns, nt]} - \delta) \Rightarrow
W\left(tI(\delta; \delta)\right) - W\left(sI(\delta; \delta) \right)
\] in $D[s + \tau, 1]$ where $W(\cdot)$ is a Wiener process, and
\[ 
I(\lambda; \delta) = \sum_{i = 1}^\infty \frac{p_{>i}(\delta)}{(i + \lambda)^2} - \frac{1}{(2 + \lambda)^2},
\]
for $\lambda \in [\eta, K]$.
\end{lemma}

We defer the proof of Lemma~\ref{lem:MLE} to Section~\ref{subsec:pf_mle}, and proceed with the proof of Theorem~\ref{thm:nullhyp}.

\begin{proof}
From \eqref{LR}, we see that
\begin{align*}
-2 \log \Lambda_{nt} &= \left(2\ell_{(0, nt]}(\hat{\delta}_{(0, nt]}) - 2\ell_{(0, nt]}(\hat{\delta}_n)  \right)+
\left( 2\ell_{(nt, n]}(\hat{\delta}_{(nt, n]}) - 2\ell_{(nt, n]}(\hat{\delta}_n)\right) \\
& := I_{(0, nt]} + I_{(nt, n]}.
\end{align*}
We first focus on $I_{(0, nt]}$. Applying Taylor's expansion to $\ell_{(0, nt]}(\cdot)$, we have that there exists $\delta^\star_n$ between $\hat{\delta}_{(0, nt]}$ and $\hat{\delta}_n$ such that
\begin{align*}
\ell_{(0, nt]}(\hat\delta_n) - \ell_{(0, nt]}(\hat{\delta}_{(0, nt]}) &= u_{(0, nt]}(\hat{\delta}_{(0, nt]})\left(\hat\delta_n -  \hat{\delta}_{(0, nt]} \right) + \frac{u'_{(0, nt]}(\delta^\star_n)}{2}\left(\hat\delta_n -  \hat{\delta}_{(0, nt]} \right)^2 \\
&= \frac{u'_{(0, nt]}(\delta^\star_n)}{2}\left(\hat\delta_n -  \hat{\delta}_{(0, nt]} \right)^2.
\end{align*}
Therefore,
\begin{align*}
I_{(0, nt]} = -u'_{(0, nt]}(\delta^\star_n)\left(\hat{\delta}_n -  \hat{\delta}_{(0, nt]} \right)^2 = -\frac{1}{n}u'_{(0, nt]}(\delta^\star_n)\left( \sqrt{n} \left(\hat{\delta}_n -  \hat{\delta}_{(0, nt]} \right) \right)^2.
\end{align*}
Also, since
\begin{equation}
\sqrt{n} \left(\hat{\delta}_n -  \hat{\delta}_{(0, nt]} \right) = \sqrt{n} \left(\hat{\delta}_n - \delta \right) - \sqrt{n} \left(  \hat{\delta}_{(0, nt]} - \delta \right),
\label{eq:MLEfunctional}
\end{equation}
and by Lemma~\ref{lem:MLE} as well as the continuity of the functional 
$F: x \in D[\gamma, 1] \mapsto x(1) - x$,
 we have 
\begin{equation}
\label{eq:MLEfunctional2}
t I(\delta; \delta) \cdot \sqrt{n} \left(\hat{\delta}_n -  \hat{\delta}_{(0, nt]} \right) \Rightarrow tW\left( I(\delta; \delta)\right) - W\left( tI(\delta; \delta) \right)\qquad \text{in }D[\gamma,1].
\end{equation}
Similar arguments as in (ii) and (iii) in the proof of Lemma \ref{lem:MLE} give that
\[
tI(\delta; \delta)  \cdot I_{(0, nt]} \Rightarrow\left(W\left( tI(\delta; \delta) \right) - tW(I(\delta; \delta))\right)^2,\qquad \text{in }D[\gamma,1].
\]
Defining the Brownian bridge $\displaystyle B(t):= I^{-1/2}(\delta; \delta) \left(W(tI(\delta; \delta)) - tW(I(\delta; \delta)) \right)$, we thus have
\[
I_{(0, nt]} \Rightarrow \frac{B^2(t)}{t},\qquad \text{in }D[\gamma,1].
\]
Nearly the same proof methodology applies for $I_{(nt, n]}$, thus giving that
\begin{equation}
\label{eq:I}
I_{(nt, n]} \Rightarrow \frac{B^2(t)}{1 - t}
 \qquad \text{in }D[0,1 - \gamma].
\end{equation}
Applying the continuity of the functional $G: x(t) \in D[\gamma, 1 - \gamma]  \mapsto x(t) + x(1 - t)$ we thus have that $I_{(0, nt]} + I_{(nt, n]}$ converges weakly to the sum of their respective limits, which can be rewritten as 
\begin{align*}
\frac{B^2(t)}{t} + \frac{B^2(t)}{1 - t} = \frac{B^2(t)}{t(1 - t)},
\end{align*}
Thus
\begin{equation}
-2 \log \Lambda_{nt} \Rightarrow \frac{B^2(t)}{t(1 - t)},\qquad \text{in }D[\gamma, 1 - \gamma].
\end{equation}
Finally, using continuity of the functional 
$$
G: x \in D[\gamma, 1 - \gamma]  \mapsto \sup_{t \in [\gamma, 1 - \gamma]} x(t),
$$ 
we have that
\begin{align*}
\sup_{t \in [\gamma, 1 - \gamma]} -2 \log \Lambda_{nt} \Rightarrow \sup_{t \in [\gamma, 1 - \gamma]} \frac{B^2(t)}{t(1 - t)} \qquad \text{in }\mathbb{R}.
\end{align*}
\end{proof}

\subsection{Proof of Lemma~\ref{lem:MLE}}\label{subsec:pf_mle}

\begin{proof}
In the spirit of \cite{wan2017fitting}, we use Taylor expansion to rewrite the score at $\hat{\delta}_{(ns, nt]}$ around $\delta$. Then
\begin{equation}
\label{eq:Taylor}
0 = u_{(ns, nt]}(\delta) + u'_{(ns, nt]}(\delta^\star_n)(\hat{\delta}_{(ns, nt]} - \delta),
\end{equation}
where $\delta^\star_n = \delta + \xi(\hat{\delta}_{(ns, nt]} - \delta)$ for some $\xi \in [0, 1]$. The intermediate value $\delta^\star_n$ depends on $t$, though we drop this dependence in notation. Some algebra then gives
\begin{equation}
\sqrt{n}(\hat{\delta}_{(ns, nt]} - \delta) = - \left(\frac{1}{\frac{1}{n} u'_{(ns, nt]}(\delta^\star_n)} \right)n^{-1/2} u_{(ns, nt]}(\delta).
\end{equation}
The following results are used to prove Lemma \ref{lem:MLE}:
\begin{itemize}
\item[(i)] $n^{-1/2} u_{(0, nt]}(\delta) \Rightarrow  W \left( tI(\delta; \delta) \right) \qquad \text{in } D[0, 1]$,
\item[(ii)] $\displaystyle \sup_{t \in [0, 1]} \sup_{\lambda \in [\eta, K]} \left| -n^{-1}u'_{(0, nt]}(\lambda) - tI(\lambda; \delta) \right| \xrightarrow{p} 0$,
\item[(iii)] $\displaystyle \sup_{t \in [s + \tau, 1]}  \left| \hat{\delta}_{(ns, nt]} - \delta \right| \xrightarrow{p} 0$.
\end{itemize}
where $W(\cdot)$ is a Wiener process on $[0, 1]$. The use of $\tau$ in (iii) ensures that $t$ is not chosen to be exactly $s$, and retains compactness of the index set. We prove these results in Lemmas \ref{score_asymp}, \ref{lem:score_prime_asymp} and \ref{lem:delta_unif}, respectively.  Lemma \ref{score_asymp} combined with the continuity of the functional $F: x \in D[s + \tau, 1] \mapsto x - x(s)$ gives that
\begin{equation}
\label{eq:W}
n^{-1/2} u_{(ns, nt]}(\delta) = n^{-1/2} u_{(0, nt]}(\delta) - n^{-1/2} u_{(0, ns]}(\delta) \Rightarrow W(tI(\delta; \delta)) - W(sI(\delta; \delta)),
\end{equation}
in $D[s + \tau, 1]$. Lemma \ref{lem:score_prime_asymp} gives that
\begin{align*}
&\sup_{t \in [s + \tau, 1]} \sup_{\lambda \in [\eta, K]} \left| -n^{-1}u'_{(ns, nt]}(\lambda) - (t - s)I(\lambda; \delta) \right| \\
&= \sup_{t \in [s + \tau, 1]} \sup_{\lambda \in [\eta, K]} \left| -n^{-1}u'_{(0, nt]}(\lambda) + n^{-1}u'_{(0, ns]}(\lambda) - tI(\lambda; \delta) + s I(\lambda; \delta) \right| \\
&\leq \sup_{t \in [s + \tau, 1]} \sup_{\lambda \in [\eta, K]} \left| -n^{-1}u'_{(0, nt]}(\lambda) - tI(\lambda; \delta) \right|+  \sup_{\lambda \in [\eta, K]}  \left| n^{-1}u'_{(0, ns]}(\lambda) + s I(\lambda; \delta) \right| \xrightarrow{p} 0.
\end{align*}
We now establish the convergence $-n^{-1}u'_{(ns, nt]}(\delta_n^\star) \xrightarrow{p} (t-s)I(\delta; \delta)$ in $D[s + \tau, 1]$. See that
\begin{align}
\label{eq:int_conv}
\begin{split}
&\sup_{t \in [s + \tau, 1]} \left| -n^{-1}u'_{(ns, nt]}(\delta_n^\star) - (t - s)I(\delta; \delta) \right| \\
&\leq \sup_{t \in [s + \tau, 1]} \left| -n^{-1}u'_{(ns, nt]}(\delta_n^\star) - (t - s)I(\delta_n^\star) \right| + \sup_{t \in [s + \tau, 1]} \left| (t - s)I(\delta_n^\star) - (t - s)I(\delta; \delta) \right| \\
&\leq \sup_{t \in [s + \tau, 1]} \sup_{\lambda \in [\eta, K]} \left| -n^{-1}u'_{(ns, nt]}(\lambda) - (t - s)I(\lambda; \delta) \right| + (1 - s) \sup_{t \in [s + \tau, 1]} \left|I(\delta_n^\star) - I(\delta; \delta) \right|.
\end{split}
\end{align}
The first term converges in probability to $0$. Since $I(\lambda; \delta)$ is (uniformly) continuous on $[\eta, K]$, it suffices to show that $\delta_n^\star \xrightarrow{p} \delta$ in $D[s + \tau, 1]$. Since $\left| \delta_n^\star - \delta \right| \leq | \hat{\delta}_{(ns, nt]} - \delta |$, we can instead show that $\sup_{t \in [s + \tau, 1]}| \hat{\delta}_{(ns, nt]} - \delta | \xrightarrow{p} 0$. This is proven in Lemma \ref{lem:delta_unif}. 

Thus, $\sup_{t \in [s + \tau, 1]} \left|I(\delta_n^\star) - I(\delta; \delta) \right| \xrightarrow{p} 0$. Hence, by \eqref{eq:int_conv}
\begin{equation}
\label{eq:fisher_D}
-n^{-1}u'_{(ns, nt]}(\delta_n^\star) \xrightarrow{p} (t-s)I(\delta; \delta) \qquad \text{in } D[s + \tau, 1].
\end{equation}
Combining the convergences \eqref{eq:W} and \eqref{eq:fisher_D}, we have that jointly in $D[s + \tau, 1] \times D[s + \tau, 1]$
\begin{align*}
\begin{pmatrix}
n^{-1/2} u_{(ns, nt]}(\delta) \\
-n^{-1}u'_{(ns, nt]}(\delta_n^\star)
\end{pmatrix} 
\Rightarrow 
\begin{pmatrix}
W(tI(\delta; \delta)) - W(sI(\delta; \delta))\\
(t-s)I(\delta; \delta)
\end{pmatrix}.
\end{align*}
Finally, by one last application of the continuous mapping theorem
\begin{align*}
(t-s)I(\delta; \delta) \cdot \sqrt{n}(\hat{\delta}_{(ns, nt]} - \delta) &= - \frac{(t-s)I(\delta; \delta)}{\frac{1}{n} u'_{(ns, nt]}(\delta^\star_n)} n^{-1/2} u_{(ns, nt]}(\delta) \\
&\Rightarrow W(tI(\delta; \delta)) - W(sI(\delta; \delta)).
\end{align*}
\end{proof}

\subsection{Supporting Lemmas for the Proof of Lemma~\ref{lem:MLE}}
\begin{lemma}
\label{score_asymp}
Suppose the graph sequence $\{G(k)\}_{k = 2}^{n}$ evolves according to PA$(\delta)$. Then as $n \rightarrow \infty$,
\[ n^{-1/2} u_{(0, nt]}(\delta) \Rightarrow W\left( tI(\delta; \delta) \right), \]
in $D[0, 1]$.
\end{lemma}
\begin{proof}
We look to apply Theorem 2.5 in \cite{durrett1978functional}, which we record as Theorem \ref{thm:FMCLT} for convenience. Towards this end, define
\begin{equation}
\eta_{n, k} = \frac{1}{\sqrt{n}}\left( \frac{1}{D_{v_k}(k - 1) + \delta} - \frac{1}{2 + \delta} \right),
\end{equation}
and let $\mathcal{F}_{n,k} = \sigma(G(0), G(1), \dots, G(k))$. $\eta_{n, k}$ is a martingale difference array since 
\begin{align*}
\mathbb{E}\left[ \eta_{n, k}  \mid  \mathcal{F}_{n,k} \right] &= \frac{1}{\sqrt{n}}\left( \sum_{w = 1}^{k - 1} \frac{1}{D_w(k - 1) + \delta} \frac{D_w(k - 1) + \delta}{(k - 1)(2 + \delta)}  - \frac{1}{2 + \delta} \right) \\
&= \frac{1}{\sqrt{n}}\left( \sum_{w = 1}^{k - 1} \frac{1}{(k - 1)(2 + \delta)}  - \frac{1}{2 + \delta} \right) = 0.
\end{align*}
To show convergence, we confirm (a) of Theorem \ref{thm:FMCLT}. Note that
\begin{align*}
&\sum_{k = 3}^{\floor{nt}} \mathbb{E}\left[ \eta^2_{n, k} \bigm\vert \mathcal{F}_{n,k} \right] \\
&= \frac{1}{n} \sum_{k = 3}^{\floor{nt}} \mathbb{E}\left[ \frac{1}{D_{v_k}(k - 1) + \delta}\left(\frac{1}{D_{v_k}(k - 1) + \delta} - \frac{1}{2 + \delta} \right) \right.\\
&\left.\qquad \qquad\qquad- \frac{1}{2 + \delta}\left(\frac{1}{D_{v_k}(k - 1) + \delta} - \frac{1}{2 + \delta} \right) \biggm\vert \mathcal{F}_{n,k} \right].
\end{align*}
The right hand term is $0$ by a similar computation as before. Continuing with the remaining term, we have
\begin{align*}
&\frac{1}{n} \sum_{k = 3}^{\floor{nt}} \mathbb{E}\left[ \frac{1}{D_{v_k}(k - 1) + \delta}\left(\frac{1}{D_{v_k}(k - 1) + \delta} - \frac{1}{2 + \delta} \right) \biggm\vert \mathcal{F}_{n,k} \right] \\
&= \frac{1}{n} \sum_{k = 3}^{\floor{nt}} \mathbb{E}\left[ \frac{1}{(D_{v_k}(k - 1) + \delta)^2} - \frac{1}{2 + \delta} \frac{1}{D_{v_k}(k - 1) + \delta} \biggm\vert\ \mathcal{F}_{n,k} \right] \\
&= \frac{1}{n} \sum_{k = 3}^{\floor{nt}} \sum_{w = 1}^{k - 1} \frac{1}{(D_{w}(k - 1) + \delta)^2} \frac{D_{w}(k - 1) + \delta}{(k - 1)(2 + \delta)} \\
&\qquad- \frac{1}{2 + \delta} \frac{1}{n} \sum_{k = 3}^{\floor{nt}}  \sum_{w = 1}^{k - 1} \frac{1}{D_{w}(k - 1) + \delta} \frac{D_{w}(k - 1) + \delta}{(k - 1)(2 + \delta)} \\
&= \frac{1}{2 + \delta} \frac{1}{n}\sum_{k = 3}^{\floor{nt}} \frac{1}{k - 1}  \sum_{w = 1}^{k - 1} \frac{1}{D_{w}(k - 1) + \delta} -  \frac{1}{2 + \delta} \frac{1}{n} \sum_{k = 3}^{\floor{nt}} \frac{1}{2 + \delta} \\
&= \frac{1}{2 + \delta} \frac{1}{n}\sum_{k = 3}^{\floor{nt}} \frac{1}{k - 1}  \sum_{w = 1}^{k - 1} \sum_{i = 1}^\infty \frac{1_{\{ D_w(k - 1) = i \}}}{i + \delta} - \frac{\floor{nt} - 2}{n}  \frac{1}{(2 + \delta)^2} \\
&= \frac{1}{2 + \delta} \frac{1}{n}\sum_{k = 3}^{\floor{nt}} \sum_{i = 1}^\infty \frac{N_i(k - 1)/(k - 1)}{i + \delta} - \frac{\floor{nt} - 2}{n}  \frac{1}{(2 + \delta)^2} \\
&= \frac{t}{2 + \delta}\frac{1}{nt} \sum_{k = 3}^{\floor{nt}} \sum_{i = 1}^\infty \frac{N_i(k - 1)/(k - 1)}{i + \delta} - \frac{\floor{nt} - 2}{n}  \frac{1}{(2 + \delta)^2}.
\end{align*}
Since $N_i(n)/n \xrightarrow{a.s.} p_i(\delta)$ by \eqref{deg_counts}, Ces\`aro convergence of random varibles gives that $\frac{1}{nt} \sum_{k = 1}^{\floor{nt}} \sum_{i = 1}^\infty \frac{N_i(k - 1)/(k - 1)}{i + \delta} \xrightarrow{a.s.}  \sum_{i = 1}^\infty \frac{p_i(\delta)}{i + \delta} = \frac{1}{2 + \delta} \sum_{i = 1}^\infty \frac{p_{>i}(\delta)}{(i + \delta)^2}$. Thus
\begin{equation}
\sum_{k = 3}^{\floor{nt}} \mathbb{E}\left[ \eta^2_{n, k} \bigm\vert \ \mathcal{F}_{n,k} \right] \xrightarrow{p}  t\left(\sum_{i = 1}^\infty \frac{p_{>i}(\delta)}{(i + \delta)^2} - \frac{1}{(2 + \delta)^2}\right) = tI(\delta; \delta).
\end{equation}
Next we confirm (b) of Theorem \ref{thm:FMCLT}. Fix $\epsilon > 0$ and note that 
\begin{equation}
\left| \frac{1}{D_{v_k}(k - 1) + \delta} - \frac{1}{2 + \delta} \right| \leq \left| \frac{1}{D_{v_k}(k - 1) + \delta}\right| + \left|\frac{1}{2 + \delta} \right| \leq  \frac{1}{1 + \delta}  +  \frac{1}{2 + \delta}  := m(\delta).
\label{bound}
\end{equation}
Hence, $|\eta_{n, k}| \leq \frac{1}{\sqrt{n}}m(\delta)$. Thus
\begin{align*}
\sum_{k = 3}^{\floor{nt}} \mathbb{E}\left[ \eta^2_{n, k} 1_{\{|\eta_{n, k}| > \epsilon \}} \bigm\vert \mathcal{F}_{n,k} \right] \leq m^2(\delta) \frac{1}{n} \sum_{k = 3}^{\floor{nt}} \mathbb{P}\left(|\eta_{n, k}| > \epsilon  \bigm\vert \mathcal{F}_{n,k}\right).
\end{align*}
Since $\mathbb{P}\left(|\eta_{n, k}| > \epsilon  \bigm\vert \mathcal{F}_{n,k}\right) \xrightarrow{a.s.} 0$, Ces\`aro convergence thus gives that the bound tends to $0$ almost surely and $\sum_{k = 3}^{\floor{nt}} \mathbb{E}\left[ \eta^2_{n, k} 1_{\{|\eta_{n, k}| > \epsilon \}} \bigm\vert \mathcal{F}_{n,k} \right]  \xrightarrow{a.s.} 0$. Since all conditions are satisfied, Theorem \ref{thm:FMCLT} gives
\begin{equation}
n^{-1/2} u_{(0, nt]}(\delta) = \frac{1}{\sqrt{n}} \sum_{k = 3}^{\floor{nt}} \left( \frac{1}{D_{v_k}(k - 1) + \delta} - \frac{1}{2 + \delta} \right) \Rightarrow W\left( tI(\delta; \delta) \right) \text{ in } D[0, 1].
\end{equation}
\end{proof}

\begin{theorem}
\label{thm:FMCLT}
Suppose $\left\lbrace X_{n, i}, \mathcal{F}_{n, i}  \right\rbrace$ is a martingale difference array where $\mathcal{F}_{n, i}$ is sequence of sigma algebras that increase with $n$. If
\begin{itemize}
\item[(a)] for all $t \in [0, 1]$, $\sum_{i = 1}^{\floor{nt}} \mathbb{E} \left[X_{n, i} \mid  \mathcal{F}_{n, i} \right] \xrightarrow{p} \varphi(t)$ where $\varphi$ is continuous and
\item[(b)] for all $\epsilon > 0$, $\sum_{i = 1}^{\floor{nt}} \mathbb{E} \left[X^2_{n, i} 1_{\left\lbrace |X_{n, i}| > \epsilon\right\rbrace} \mid  \mathcal{F}_{n, i} \right] \xrightarrow{p} 0$
\end{itemize}
then $\sum_{i = 1}^{\floor{nt}} X_{n, i} \Rightarrow W(\varphi(t))$ in $D[0, 1]$, where $W(\cdot)$ is a Wiener process. 
\end{theorem}

\begin{lemma}
\label{lem:score_prime_asymp}
Suppose the graph sequence $\{G(k)\}_{k = 2}^{n}$ evolves according to PA$(\delta)$. Then as $n \rightarrow \infty$,
\[ \sup_{t \in [0, 1]} \sup_{\lambda \in [\eta, K]} \left| n^{-1}u_{(0, nt]}(\lambda) - tU(\lambda; \delta) \right| \xrightarrow{p} 0, \]
and
\[ \sup_{t \in [0, 1]} \sup_{\lambda \in [\eta, K]} \left| -n^{-1}u'_{(0, nt]}(\lambda) - tI(\lambda; \delta) \right| \xrightarrow{p} 0, \]
where 
\[ U(\lambda; \delta) =  \sum_{i = 1}^\infty \frac{p_{>i}(\delta)}{i + \lambda} - \frac{1}{2 + \lambda}.\]
\end{lemma}
\begin{proof}
We only prove the second result, as the proof for the first is nearly identical. The proof follows those of \cite{gao2017asymptotic} and \cite{wan2017fitting}, with minor adjustments for the supremum norm in $t$. First, see that by using a similar strategy as in Lemma \ref{score_asymp}
\begin{align*}
-\frac{1}{n} u'_{(0, nt]}(\lambda) &= \frac{1}{n} \sum_{k = 3}^{\floor{nt}} \frac{1}{(D_{v_k}(k - 1) + \lambda)^2} - \frac{\floor{nt} - 2}{n(2 + \lambda)^2} \\
&= \frac{1}{n} \sum_{k = 3}^{\floor{nt}} \sum_{i = 1}^\infty \frac{1_{\{ D_{v_k}(k - 1) = i \}}}{(i + \lambda)^2} - \frac{\floor{nt} - 2}{n(2 + \lambda)^2} \\
&= \sum_{i = 1}^\infty \frac{N_{>i}(\floor{nt})/n}{(i + \lambda)^2} -  \frac{1}{n} \sum_{i = 1}^\infty \frac{21_{\{i = 1 \}}}{(i + \lambda)^2} - \frac{\floor{nt} - 2}{n(2 + \lambda)^2} \\
&= \sum_{i = 1}^\infty \frac{N_{>i}(\floor{nt})/n}{(i + \lambda)^2} - \frac{1}{n}  \frac{2}{(1 + \lambda)^2} - \frac{\floor{nt} - 2}{n(2 + \lambda)^2}.
\end{align*}
Further,
\begin{align*}
\sup_{t \in [0, 1]}& \sup_{\lambda \in [\eta, K]} \left| -n^{-1}u'_{(0, nt]}(\lambda) - tI(\lambda; \delta) \right| \\
&\leq \sup_{t \in [0, 1]} \sup_{\lambda \in [\eta, K]} \left| \sum_{i = 1}^\infty \frac{N_{>i}(\floor{nt})/n}{(i + \lambda)^2} - \sum_{i = 1}^\infty \frac{tp_{>i}(\delta)}{(i + \lambda)^2}\right| + \frac{1}{n}  \frac{2}{(1 + \eta)^2} \\
&\quad + \sup_{t \in [0, 1]} \sup_{\lambda \in [\eta, K]} \left|- \frac{\floor{nt} - 2}{n(2 + \lambda)^2} + \frac{1}{(2 + \lambda)^2} \right|.
\end{align*}
Clearly, the second term converges to $0$. The third term also converges to $0$ since 
\begin{align*}
\sup_{t \in [0, 1]} \sup_{\lambda \in [\eta, K]} \left|- \frac{\floor{nt} - 2}{n(2 + \lambda)^2} + \frac{1}{(2 + \lambda)^2} \right| &\leq \frac{1}{(2 + \eta)^2} \sup_{t \in [0, 1]} \left|- \frac{\floor{nt} - 2}{n} + 1 \right| \rightarrow 0.
\end{align*}
We now turn our attention to the first term. From \eqref{deg_counts2}, we have that $N_{>i}(\floor{nt})/n \xrightarrow{p} tp_{>i}(\delta)$ for any $t \in [0, 1]$. Since $N_{>i}(\floor{nt})/n$ is a monotone and  $tp_{>i}(\delta)$ is continuous as functions of $t$, this convergence is uniform, i.e.
\begin{equation}
\label{eq:unifdeg}
\sup_{t \in [0, 1]}\left|N_{>i}(\floor{nt})/n -  tp_{>i}(\delta) \right| \xrightarrow{p} 0.
\end{equation}
Also note that for every $i \geq 1$
\begin{align*}
iN_{>i}(\floor{nt}) = i \sum_{j = i + 1}^\infty N_{j}(\floor{nt}) \leq \sum_{j = 1}^\infty jN_{j}(\floor{nt}) = 2 \floor{nt},
\end{align*}
since summing over the degrees returns twice the number of edges. Hence 
\begin{equation*}
N_{>i}(\floor{nt})/nt \leq N_{>i}(\floor{nt})/\floor{nt} \leq 2/i.
\end{equation*}
We thus have for any $M \in \mathbb{N}$
\begin{align*}
&\sup_{t \in [0, 1]}  \sup_{\lambda \in [\eta, K]}  \left| \sum_{i = 1}^\infty \frac{N_{>i}(\floor{nt})/n}{(i + \lambda)^2} - \sum_{i = 1}^\infty \frac{tp_{>i}(\delta)}{(i + \lambda)^2}\right| \\
&\leq   \sum_{i = 1}^\infty \frac{\sup_{t \in [0, 1]} \left| N_{>i}(\floor{nt})/n - tp_{>i}(\delta) \right|}{(i + \eta)^2} \\
&\leq \sum_{i = 1}^M \frac{\sup_{t \in [0, 1]} \left| N_{>i}(\floor{nt})/n - tp_{>i}(\delta) \right|}{(i + \eta)^2} + \sum_{i = M + 1}^\infty \frac{2/i}{(i + \eta)^2} + \sum_{i = M + 1}^\infty \frac{p_{>i}(\delta)}{(i + \eta)^2}.
\end{align*}
By choosing large enough $M$, the last two terms can be made arbitrarily small. By \eqref{eq:unifdeg}, the first term converges in probability to $0$. Putting these results together, we have
\[\sup_{t \in [0, 1]} \sup_{\lambda \in [\eta, K]} \left| -n^{-1}u'_{(0, nt]}(\lambda) - tI(\lambda; \delta) \right| \xrightarrow{p} 0. \]
\end{proof}

\begin{lemma}
\label{lem:delta_unif}
Fix $s \in [0, 1]$ and $\tau \in (0, 1 - s)$. Assume that $\{G(k)\}_{k = 2}^{n}$ evolves according to the PA$(\delta)$ rule.
Then as $n \rightarrow \infty$
\begin{align*}
\sup_{t \in [s + \tau, 1]} \left| \hat{\delta}_{(ns, nt]} - \delta \right| \xrightarrow{p} 0.
\end{align*}
\end{lemma}
\begin{proof}
This proof closely follows that of Theorem 3.2 in \cite{wan2017fitting}. From Lemma 4 of \cite{gao2017asymptotic}, $U(\lambda; \delta)$ has a unique zero at $\lambda = \delta$. Furthermore, it is positive for $\lambda < \delta$ and negative for $\lambda > \delta$. 

Fix $\epsilon > 0$. By continuity of $U(\lambda; \delta)$, there exits $\xi > 0$ such that $U(\lambda; \delta) > \epsilon/\tau$ on $[\eta, \delta - \xi]$ and $U(\lambda; \delta) < -\epsilon/\tau$ on $[\delta + \xi, K]$. Hence
\begin{align*}
\inf_{\lambda \in [\eta, \delta - \xi]} \inf_{ t \in [s + \tau, 1]}(t - s)U(\lambda; \delta) > \epsilon  \qquad \text{and} \qquad  \sup_{\lambda \in [\delta + \xi, K]} \sup_{ t \in [s + \tau, 1]}(t - s)U(\lambda; \delta) < -\epsilon .
\end{align*}
Further, for all $\lambda \in [\eta, \delta - \xi]$ and $t \in [s + \tau, 1]$,
\begin{align*}
n^{-1}u_{(ns, nt]}(\lambda) &\geq (t - s)U(\lambda; \delta) - \sup_{v \in [s + \tau, 1]} \sup_{\lambda \in [\eta, K]} \left| n^{-1}u_{(ns, nv]}(\lambda) - (v - s)U(\lambda; \delta) \right| \\
&> \epsilon - \sup_{v \in [s + \tau, 1]} \sup_{\lambda \in [\eta, K]} \left| n^{-1}u_{(ns, nv]}(\lambda) - (v - s)U(\lambda; \delta) \right|.
\end{align*}
Similarly, for all $\lambda \in [\delta + \xi, K]$ and $t \in [s + \tau, 1]$,
\begin{align*}
n^{-1}u_{(ns, nt]}(\lambda) &\leq (t - s)U(\lambda; \delta) + \sup_{v \in [s + \tau, 1]} \sup_{\lambda \in [\eta, K]} \left| n^{-1}u_{(ns, nv]}(\lambda) - (v - s)U(\lambda; \delta) \right| \\
&< -\epsilon + \sup_{v \in [s + \tau, 1]} \sup_{\lambda \in [\eta, K]} \left| n^{-1}u_{(ns, nv]}(\lambda) - (v - s)U(\lambda; \delta) \right|.
\end{align*}
By Lemma \ref{lem:score_prime_asymp}, we see that
\begin{align*}
\mathbb{P} \left( \sup_{v \in [s + \tau, 1]} \sup_{\lambda \in [\eta, K]} \left| n^{-1}u_{(ns, nv]}(\lambda) - (v - s)U(\lambda; \delta) \right| < \epsilon/2 \right) \rightarrow 1 \qquad \text{as } n\rightarrow \infty. 
\end{align*}
Hence with probability tending towards $1$ as $n \rightarrow \infty$:
\begin{align*}
&\inf_{t \in [s + \tau, 1]} \inf_{\lambda \in [\eta, \delta - \xi]} n^{-1}u_{(ns, nt]}(\lambda) \geq \epsilon/2,\\
& \sup_{t \in [s + \tau, 1]} \sup_{\lambda \in [\delta + \xi, K]} n^{-1}u_{(ns, nt]}(\lambda) \leq -\epsilon/2.
\end{align*}
Since $u_{(ns, nt]}(\hat{\delta}_{(ns, nt]}) = 0$ $\forall t \in [s + \tau, 1]$, this implies 
\begin{align*}
\mathbb{P}\left( \sup_{t \in [s + \tau, 1]} \left| \hat{\delta}_{(ns, nt]} - \delta \right| < \xi \right) \rightarrow 1.
\end{align*}
\end{proof}

\subsection{Proof of Theorem \ref{thm:delta_unif_chng}}\label{sec:delta_unif_chng_proof}

The proof is very similar to that of Lemma \ref{lem:delta_unif}. We replace the use of Lemma 4 in \cite{gao2017asymptotic} with Lemma \ref{lem:unique_zero}, which ensures that $U^{\star}_t(\lambda) - U^{\star}_s(\lambda)$ has a unique zero at $\delta_2$. Additionally, we replace  Lemma \ref{lem:score_prime_asymp} with Lemma \ref{lem:score_asymp_chng} to ensure the uniform convergence of $n^{-1}u_{(ns, nt]}(\lambda)$ to $U^{\star}_t(\lambda) - U^{\star}_s(\lambda)$ on $[\eta, K]$. The result follows accordingly. Lemmas \ref{lem:score_asymp_chng} and \ref{lem:unique_zero} are proven below.

\begin{lemma}
\label{lem:score_asymp_chng}
Suppose the graph sequence $\{G(k)\}_{k = 2}^{n}$ evolves according to PA$(t^\star; \delta_1, \delta_2)$. Then as $n \rightarrow \infty$,
\begin{align*}
\sup_{t \in [t^\star, 1]} \sup_{\lambda \in [\eta, K]} \left| n^{-1}u_{(0, nt]}(\lambda)  - U_t^\star(\lambda) \right| \xrightarrow{p} 0,
\end{align*}
and
\begin{align*}
\sup_{t \in [t^\star, 1]} \sup_{\lambda \in [\eta, K]} \left| -n^{-1}u'_{(0, nt]}(\lambda)  - I_t^\star(\lambda) \right| \xrightarrow{p} 0,
\end{align*}
where 
\begin{align*}
U_t^\star(\lambda) \equiv&  t \left(\sum_{i = 1}^\infty \frac{p_{>i}^\star(t; \delta_1, \delta_2)}{i + \lambda} - \frac{1}{2 + \lambda}\right),
\end{align*}
and 
\begin{align*}
I_t^\star(\lambda) \equiv&  t \left(\sum_{i = 1}^\infty \frac{p_{>i}^\star(t; \delta_1, \delta_2)}{(i + \lambda)^2} - \frac{1}{(2 + \lambda)^2}\right).
\end{align*}
\end{lemma}

\begin{proof}
Repeat the proof of Lemma \ref{lem:score_prime_asymp}, replacing the supremum over $[0, t]$ with a supremum over $[t^\star, t]$, until \eqref{eq:unifdeg} where we replace the convergence with Proposition 4.17 of \cite{bhamidi2018change} which states that, uniformly on $[t^\star, 1]$, $N_{>i}(\floor{nt})/n$ converges in probability to $tp_{>i}^\star(t; \delta_1, \delta_2)$. The rest of the proof follows exactly.
\end{proof}

The next lemma demonstrates that for $t > s \geq t^\star$, the limit of $n^{-1}u_{(ns, nt]}(\lambda)$, $U_t^\star(\lambda) - U_s^\star(\lambda)$ has a unique zero at $\delta_2$. This is a necessary ingredient in order to prove consistency of the MLE in the post-changepoint regime.

\begin{lemma}
\label{lem:unique_zero}
Suppose $t^\star \leq s < t \leq 1$. Then $U_t^\star(\lambda) - U_{s}^\star(\lambda)$ has a unique zero at $\lambda = \delta_2$.
\end{lemma}

\begin{proof}
Note that
\begin{align*}
U_{t}^\star(\lambda) - U_{s}^\star(\lambda) =&  \sum_{i = 1}^\infty \frac{tp_{>i}^\star(t; \delta_1, \delta_2) - sp_{>i}^\star(s; \delta_1, \delta_2)}{i + \lambda} - \frac{t - s}{2 + \lambda},
\end{align*}
which, using Lemma \ref{lem:recur}, we may write as
\begin{align*}
U_{t}^\star(\lambda) - U_{s}^\star(\lambda) =&  \int_s^t \left(\sum_{i = 1}^\infty \frac{p_{>i}^\star(u; \delta_1, \delta_2)(i + \delta_2)}{(i + \lambda)(2 + \delta_2)} - \frac{1}{2 + \lambda} \right) du \\
=& \int_s^t \left(\sum_{i = 1}^\infty \frac{p_{i}^\star(u; \delta_1, \delta_2)(i + \delta_2)}{(i + \lambda)(2+\delta_2)} - \sum_{i = 1}^\infty \frac{p_{i}^\star(u; \delta_1, \delta_2)(i + \lambda)}{(2 + \lambda)(i + \lambda)} \right) du \\
=& \int_s^t \left(\sum_{i = 1}^\infty \frac{p_{i}^\star(u; \delta_1, \delta_2)}{i + \lambda}\left(\frac{i + \delta_2}{2+\delta_2} -\frac{i + \lambda}{2 + \lambda} \right) \right) du \\
=& \frac{\delta_2 - \lambda}{(2 + \delta_2)(2 + \lambda)} \int_s^t \sum_{i = 1}^\infty \frac{p_{i}^\star(u; \delta_1, \delta_2)}{i + \lambda}(2 - i) du \\
\equiv& \frac{\delta_2 - \lambda}{(2 + \delta_2)(2 + \lambda)} \int_s^t J(u, \lambda) du.
\end{align*}
We prove that the integrand $J(u, \lambda)$ is strictly positive. See that
\begin{align*}
\sum_{i = 1}^\infty \frac{p_{i}^\star(u; \delta_1, \delta_2)}{i + \lambda}(2 - i) =&  \frac{p_{1}^\star(u; \delta_1, \delta_2)}{1 + \lambda} - \sum_{i = 3}^\infty \frac{p_{i}^\star(u; \delta_1, \delta_2)}{i + \lambda}(i - 2) \\
>& \frac{p_{1}^\star(u; \delta_1, \delta_2)}{1 + \lambda} - \frac{1}{1 + \lambda} \sum_{i = 3}^\infty p_{i}^\star(u; \delta_1, \delta_2)(i - 2) \\
=& \frac{1}{1 + \lambda} \sum_{i = 1}^\infty p_{i}^\star(u; \delta_1, \delta_2)(2 - i) \\
=&\frac{2 + \delta_2}{1 + \lambda}  - \frac{1}{1 + \lambda} \sum_{i = 1}^\infty (i + \delta_2)p_{i}^\star(u; \delta_1, \delta_2). 
\end{align*}
From Lemma \ref{thm:degdistchng}, we see that
\begin{align*}
\sum_{i = 1}^\infty (i + \delta_2)p_{i}^\star(u; \delta_1, \delta_2) =& (1 - (t^\star/u))\sum_{i = 1}^\infty (i + \delta_2)\mathbb{P}\left(\xi_{\delta_2}(\tilde{T}(u)) = i \right) \\
&+ (t^\star/u)\sum_{i = 1}^\infty (i + \delta_2)\mathbb{P}\left(\xi_{\delta_2}^{\xi_{\delta_1}(T)}(\tau^\star(u)) = i \right) \\
=& (1 - (t^\star/u))\mathbb{E}\left[ \xi_{\delta_2}(\tilde{T}(u)) + \delta_2\right] \\
&+ (t^\star/u)\mathbb{E}\left[ \xi_{\delta_2}^{\xi_{\delta_1}(T)}(\tau^\star(u)) + \delta_2\right].
\end{align*}
Using the fact that $\xi_{\delta_2}(t) - 1$ has a negative binomial distribution with number of successes $1 + \delta_2$ and probability of success $e^{-t}$, we may additionally compute that
\begin{align*}
\mathbb{E}\left[ \xi_{\delta_2}(\tilde{T}(u)) + \delta_2\right] =& \mathbb{E}\left[ \mathbb{E}\left[ \xi_{\delta_2}(\tilde{T}(u)) + \delta_2 \mid \tilde{T}(u) \right]\right] \\
=& \mathbb{E}\left[ (1 + \delta_2)e^{\tilde{T}(u)}\right].
\end{align*}
A straight-forward calculation shows that
\begin{align*}
\mathbb{E}\left[ e^{\tilde{T}(u)} \right] = \frac{2 + \delta_2}{1 + \delta_2} \frac{1 - (t^\star/u)^{\frac{1 + \delta_2}{2 + \delta_2}}}{1 - t^\star/u},
\end{align*}
so that in totality we may write 
\begin{align}
\label{eq:zeroone}
\mathbb{E}\left[ \xi_{\delta_2}(\tilde{T}(u)) + \delta_2\right] =& (2 + \delta_2) \frac{1 - (t^\star/u)^{\frac{1 + \delta_2}{2 + \delta_2}}}{1 - t^\star/u}. 
\end{align}
Similarly, we may condition on $\xi_{\delta_1}(T)$ to write
\begin{align*}
\mathbb{E}\left[ \xi_{\delta_2}^{\xi_{\delta_1}(T)}(\tau^\star(u)) + \delta_2\right] =& \mathbb{E}\left[ \mathbb{E}\left[ \xi_{\delta_2}^{\xi_{\delta_1}(T)}(\tau^\star(u)) + \delta_2 \mid \xi_{\delta_1}(T)\right]\right] \\
=&  \mathbb{E}\left[ (u/t^\star)^{\frac{1}{2 + \delta_2}}(\xi_{\delta_1}(T) + \delta_2) \right].
\end{align*}
Recalling \eqref{eq:alt1}, we have that
\begin{align*}
\mathbb{E}\left[\xi_{\delta_1}(T) + \delta_1 \right] =& (2 + \delta_1)\frac{\Gamma(3 + 2\delta_1)}{\Gamma(1 + \delta_1)} \sum_{i = 1}^\infty (i + \delta_1)\frac{\Gamma(i + \delta_1)}{\Gamma(i + 3 + 2\delta_1)} \\
=& (2 + \delta_1)\frac{\Gamma(3 + 2\delta_1)}{\Gamma(1 + \delta_1)} \sum_{i = 1}^\infty  \frac{\Gamma(i + 1 + \delta_1)}{\Gamma(i + 3 + 2\delta_1)} \\
=& (2 + \delta_1)\frac{\Gamma(3 + 2\delta_1)}{\Gamma(1 + \delta_1)} \sum_{i = 1}^\infty \frac{1}{1 + \delta_1} \left( \frac{\Gamma(i + 1 + \delta_1)}{\Gamma(i + 2 + 2\delta_1)} -  \frac{\Gamma(i + 2 + \delta_1)}{\Gamma(i + 3 + 2\delta_1)} \right) \\
=& (2 + \delta_1)\frac{\Gamma(3 + 2\delta_1)}{\Gamma(1 + \delta_1)} \frac{1}{1 + \delta_1}  \frac{\Gamma(2 + \delta_1)}{\Gamma(3 + 2\delta_1)} \\
=& 2 + \delta_1,
\end{align*}
so that $\mathbb{E}[\xi_{\delta_1}(T) + \delta_2] = 2 + \delta_2$. Thus, we have calculated that
\begin{align}
\label{eq:zerotwo}
\mathbb{E}\left[ \xi_{\delta_2}^{\xi_{\delta_1}(T)}(\tau^\star(u)) + \delta_2 \right] = (u/t^\star)^{\frac{1}{2 + \delta_2}}(2 + \delta_2). 
\end{align}
Using \eqref{eq:zeroone} and \eqref{eq:zerotwo}, we may write
\begin{align*}
\sum_{i = 1}^\infty (i + \delta_2)p_{i}^\star(u; \delta_1, \delta_2) =&  (1 - (t^\star/u))\mathbb{E}\left[ \xi_{\delta_2}(\tilde{T}(u)) + \delta_2\right] \\
&+ (t^\star/u)\mathbb{E}\left[ \xi_{\delta_2}^{\xi_{\delta_1}(T)}(\tau^\star(u)) + \delta_2\right] \\
=& (2 + \delta_2) (1 - (t^\star/u)^{\frac{1 + \delta_2}{2 + \delta_2}}) + (2 + \delta_2)(t^\star/u)^{\frac{1 + \delta_2}{2 + \delta_2}} \\
=& 2 + \delta_2.
\end{align*}
Returing to our original goal, we have shown that
\begin{align*}
J(u, \lambda) =& \sum_{i = 1}^\infty \frac{p_{i}^\star(u; \delta_1, \delta_2)}{i + \lambda}(2 - i) \\
>& \frac{2 + \delta_2}{1 + \lambda}  - \frac{1}{1 + \lambda} \sum_{i = 1}^\infty (i + \delta_2)p_{i}^\star(u; \delta_1, \delta_2) \\
=& \frac{2 + \delta_2}{1 + \lambda}  - \frac{2 + \delta_2}{1 + \lambda} \\
=& 0.
\end{align*}
In summary, we have written 
\begin{align*}
U_{t}^\star(\lambda) - U_{s}^\star(\lambda) =& \frac{\delta_2 - \lambda}{(2 + \delta_2)(2 + \lambda)} \int_s^t J(u, \lambda) du,
\end{align*}
where $J(u, \lambda) > 0$. Thus, for $t > s$, $U_{t}^\star(\lambda) - U_{s}^\star(\lambda)$ is strictly positive for $\lambda < \delta_2$, strictly negative for $\lambda > \delta_2$ and exactly zero when $\lambda = \delta_2$. Hence $\delta_2$ is the unique zero of $U_{t}^\star(\lambda) - U_{s}^\star(\lambda)$.
\end{proof}

\subsection{Proof of Theorem \ref{thm:MLE_chng}}\label{sec:MLE_chng_proof}

\begin{proof}
The proof is very similar to that of Lemma \ref{lem:MLE}. Recall that a Taylor expansion gives that
\begin{align*}
\sqrt{n}(\hat{\delta}_{(ns, nt]} - \delta_2) = -\left(\frac{1}{\frac{1}{n}u'_{(ns, nt]}(\delta^\star_n)} \right)n^{-1/2}u_{(ns, nt]}(\delta_2),
\end{align*}
where $\delta^\star_n$ lies between $\hat{\delta}_{(ns, nt]}$ and $\delta_2$. In order to prove \ref{thm:MLE_chng}, it suffices to prove
\begin{itemize}
\item[(i)] $n^{-1/2} u_{(ns, nt]}(\delta_2) \Rightarrow  W \left( I_t^\star(\delta_2) - I_s^\star(\delta_2) \right) \qquad \text{in } D[0, 1]$,
\item[(ii)] $\displaystyle \sup_{t \in [t^\star, 1]} \sup_{\lambda \in [\eta, K]} \left| -n^{-1}u'_{(ns, nt]}(\lambda) - \left(I_t^\star(\lambda) - I_s^\star(\lambda)\right) \right| \xrightarrow{p} 0$,
\item[(iii)] $\displaystyle \sup_{t \in [s + \tau, 1]}  \left| \hat{\delta}_{(ns, nt]} - \delta_2 \right| \xrightarrow{p} 0$.
\end{itemize}
Statements (ii) and (iii) are proven in Lemma \ref{lem:score_asymp_chng} and Theorem \ref{thm:delta_unif_chng}, respectively. Hence, it suffices to prove (i). As in Lemma \ref{score_asymp}, we aim to apply Theorem \ref{thm:FMCLT} to the martingale differences
\begin{align*}
\eta_{n, k} = \frac{1}{\sqrt{n}}\left( \frac{1}{D_{v_k}(k - 1) + \delta_2} - \frac{1}{2 + \delta_2}\right),
\end{align*}
for $k > \floor{ns}$. In order to prove that condition (a) of Theorem \ref{thm:FMCLT} is satisfied, recall as in the proof of Lemma \ref{score_asymp} that we may write  
\begin{align*}
\sum_{k = \floor{ns} + 1}^{\floor{nt}} \mathbb{E}\left[\eta_{n, k}^2 \mid \mathcal{F}_{n, k} \right] =& \frac{1}{2 + \delta_2} \frac{1}{n} \sum_{k = \floor{ns} + 1}^{\floor{nt}} \sum_{i = 1}^\infty \frac{N_i(k-1)/(k-1)}{i + \delta_2} \\
&- \frac{\floor{nt} - \floor{ns}}{n} \frac{1}{(2 + \delta_2)^2}.
\end{align*}
Note that under the changepoint model, $N_i(\floor{nt})/nt \xrightarrow{p} p_i^\star(t; \delta_1, \delta_2)$, uniformly in $t$, as $k \rightarrow \infty$ by Theorem \ref{thm:degdistchng}. Hence we may conclude that
\begin{align*}
 \frac{1}{n} \sum_{k = \floor{ns} + 1}^{\floor{nt}} \frac{N_i(k-1)}{k - 1} \xrightarrow{p} \int_s^t p_i^\star(u; \delta_1, \delta_2) du,
\end{align*}
where Lemma \ref{lem:recur} gives that
\begin{align*}
\int_s^t p_i^\star(u; \delta_1, \delta_2) du = \frac{2 + \delta_2}{i + \delta_2} \left(tp^\star_{>i}(t; \delta_1, \delta_2) - sp^\star_{>i}(s; \delta_1, \delta_2) \right).
\end{align*}
Hence
\begin{align*}
\sum_{k = \floor{ns} + 1}^{\floor{nt}} \mathbb{E}\left[\eta_{n, k}^2 \mid \mathcal{F}_{n, k} \right] \xrightarrow{p}& \sum_{i = 1}^\infty\frac{tp_{>i}^\star(t; \delta_1, \delta_2) - sp_{>i}^\star(s; \delta_1, \delta_2)}{(i + \delta_2)^2} - \frac{t - s}{(2 + \delta_2)^2} \\
&= I_t^\star(\delta_2) - I_s^\star(\delta_2).
\end{align*}
Thus condition (a) of Theorem \ref{thm:FMCLT} is proved. Condition (b) of  Theorem \ref{thm:FMCLT} is proved just as in Theorem \ref{score_asymp}. Hence, (i) is proved and asymptotic normality of the MLE is achieved.

We now prove the variance bound. See that by using Lemma \ref{lem:recur}
\begin{align*}
I_t^\star(\delta_2) - I_s^\star(\delta_2)  =& \sum_{i = 1}^\infty\frac{tp_{>i}^\star(t; \delta_1, \delta_2) - sp_{>i}^\star(s; \delta_1, \delta_2)}{(i + \delta_2)^2} -  \frac{t - s}{(2 + \delta_2)^2}  \\
=& \frac{1}{2 + \delta_2} \int_s^t \left(\sum_{i = 1}^\infty \frac{p_i^\star(u;\delta_1, \delta_2)}{i + \delta_2} - \frac{1}{2 + \delta_2}\right) du.
\end{align*}
Note that $p_i^\star(u;\delta_2, \delta_2) = p_i(\delta_2)$. This can be easily seen with the identity
\begin{equation}
\label{eq:a} 
\mathbb{P}\left( \xi^{j+1}_\delta(t) > i \right) = \mathbb{P}\left( \xi^j_\delta(t) > i  \right) + \frac{i + \delta}{j + \delta} \mathbb{P}\left( \xi^j_\delta(t) = i \right),
\end{equation}
for $i > j \geq 1$ and $t \geq 0$. Equation \eqref{eq:a} is easily obtained through an integration by parts. With \eqref{eq:a} in hand, see that
\begin{align*}
p_i^\star(u;\delta_2, \delta_2) =& p_i(\delta_2)\mathbb{P}\left(\xi_{2\delta_2}^3(\tau^\star(u)) > i + 2 \right) \\
&+ (t^\star/u)\sum_{j = 1}^ip_j(\delta_2)\mathbb{P}\left(\xi_{\delta_2}^j(\tau^\star(u)) = i\right) \\
=& p_i(\delta_2)\mathbb{P}\left(\xi_{2\delta_2}^3(\tau^\star(u)) > i + 2 \right) \\
&+ p_i(\delta_2)\sum_{j = 1}^i \frac{i + 2 + 2\delta_2}{j + 2 + 2\delta_2}\mathbb{P}\left(\xi_{2\delta_2}^{j+2}(\tau^\star(u)) = i + 2\right), \\
\intertext{and using \eqref{eq:a}}
=&  p_i(\delta_2)\mathbb{P}\left(\xi_{2\delta_2}^3(\tau^\star(u)) > i + 2 \right) \\
&+ p_i(\delta_2)\sum_{j = 1}^i\left(\mathbb{P}\left(\xi_{2\delta_2}^{j+3}(\tau^\star(u)) > i + 2\right) - \mathbb{P}\left(\xi_{2\delta_2}^{j+2}(\tau^\star(u)) > i + 2\right) \right) \\
=&  p_i(\delta_2)\mathbb{P}\left(\xi_{2\delta_2}^3(\tau^\star(u)) > i + 2 \right) \\
&+ p_i(\delta_2)\left(1 - \mathbb{P}\left(\xi_{2\delta_2}^3(\tau^\star(u)) > i + 2 \right) \right)  \\
=& p_i(\delta_2).
\end{align*}
Hence, if  $\delta_1 = \delta_2$, $I_t^\star(\delta_2) - I_s^\star(\delta_2) = (t - s)I(\delta_2; \delta_2)$. Further, by interchanging sums we may rewrite
\begin{align*}
\sum_{i = 1}^\infty \frac{p_i^\star(u;\delta_1, \delta_2)}{i + \delta_2} =& \sum_{i = 1}^\infty \frac{p_i(\delta_2)\mathbb{P}\left(\xi_{2\delta_2}^3\left(\tau^\star(u)\right) > i + 2 \right)}{i + \delta_2} \\
&+ (t^\star/u)\sum_{i = 1}^\infty\frac{\sum_{j = 1}^ip_j(\delta_1)\mathbb{P}\left(\xi_{\delta_2}^j(\tau^\star(u)) = i\right)}{i + \delta_2} \\
=& \sum_{i = 1}^\infty \frac{p_i(\delta_2)\mathbb{P}\left(\xi_{2\delta_2}^3\left(\tau^\star(u)\right) > i + 2 \right)}{i + \delta_2} \\
&+ (t^\star/u)\sum_{j = 1}^\infty p_j(\delta_1) \sum_{i = j}^\infty\frac{\mathbb{P}\left(\xi_{\delta_2}^j(\tau^\star(u)) = i\right)}{i + \delta_2} \\
\equiv& \sum_{i = 1}^\infty \frac{p_i(\delta_2)\mathbb{P}\left(\xi_{2\delta_2}^3\left(\tau^\star(u)\right) > i + 2 \right)}{i + \delta_2} + (t^\star/u)\sum_{j = 1}^\infty p_j(\delta_1) w_j.
\end{align*}
We thus aim to show that if $\delta_1 < \delta_2$ then $\sum_{j = 1}^\infty p_j(\delta_1) w_j > \sum_{j = 1}^\infty p_j(\delta_2) w_j$ and thus $I_t^\star(\delta_2) - I_s^\star(\delta_2) > (t - s)I(\delta_2; \delta_2)$. If $\delta_1 > \delta_2$, we will also show that the inequality is reversed. We rewrite
\begin{align*}
\sum_{j = 1}^\infty p_j(\delta_1) w_j =& p_1(\delta_1)w_1 + \sum_{j = 2}^\infty p_j(\delta_1) w_j \\
=& \left(1 - \sum_{j = 2}^\infty p_j(\delta_1) \right)w_1 + \sum_{j = 2}^\infty p_j(\delta_1) w_j \\ 
=& w_1 + \sum_{j = 2}^\infty p_j(\delta_1) (w_j - w_1) \\
=& w_1 + \sum_{j = 2}^\infty p_j(\delta_1) \sum_{m = 1}^{j - 1}(w_{m + 1} - w_m) \\
=& w_1 + \sum_{m = 1}^\infty  \left(\sum_{j = m + 1}^{\infty}p_{j}(\delta_1) \right)(w_{m + 1} - w_m) \\
=& w_1 + \sum_{m = 1}^\infty  p_{>m}(\delta_1)(w_{m + 1} - w_m) \\
=& w_1 - \sum_{m = 1}^\infty  p_{>m}(\delta_1)(w_{m} - w_{m + 1}).
\end{align*}
We thus aim to prove that if $\delta_1 < \delta_2$, then $\sum_{m = 2}^\infty  p_{>m}(\delta_1)(w_{m} - w_{m + 1}) < \sum_{m = 2}^\infty  p_{>m}(\delta_2)(w_{m} - w_{m + 1})$. If $\delta_1 > \delta_2$, we also aim to prove that the inequality is reversed. In order to do so, we aim to apply Lemma 3 of \cite{gao2017asymptotic}, which requires that $\{w_{m} - w_{m + 1}\}_{m = 1}^\infty$ is a strictly decreasing, non-negative sequence. In order to prove that $\{w_{m} - w_{m + 1}\}_{m = 1}^\infty$ is non-negative,  note that $\{w_{j}\}_{j = 1}^\infty$ is a decreasing sequence since $\xi_{\delta_2}^{j+1}(\tau^\star(u))$ stochastically dominates $\xi_{\delta_2}^{j}(\tau^\star(u))$ and thus
\begin{align*}
w_{j} = \mathbb{E}\left[ \frac{1}{\xi_{\delta_2}^j(\tau^\star(u)) + \delta_2} \right] \geq \mathbb{E}\left[ \frac{1}{\xi_{\delta_2}^{j + 1}(\tau^\star(u)) + \delta_2} \right] = w_{j + 1}.
\end{align*}
We next derive the following relationship for the sequence $\{w_{j}\}_{j = 1}^\infty$:
\begin{align}
\label{eq:convex}
\left(t^\star/u\right)^{\frac{1}{2+\delta_2}}w_{j} + \left(1 - \left(t^\star/u\right)^{\frac{1}{2+\delta_2}} \right)w_{j + 1} =  \frac{\left(t^\star/u\right)^{\frac{1}{2+\delta_2}}}{j + \delta_2}.
\end{align}
See that
\begin{align*}
w_{j + 1} =& \sum_{i = j + 1}^\infty\frac{\mathbb{P}\left(\xi_{\delta_2}^{j + 1}(\tau^\star(u)) = i\right)}{i + \delta_2} \\
=& \sum_{i = j + 1}^\infty \frac{1}{i + \delta_2} \frac{\Gamma(i + \delta_2)}{\Gamma(i - j)\Gamma(j + 1 + \delta_2)} \left(\frac{t^\star}{u}\right)^{\frac{j + 1 + \delta_2}{2+\delta_2}}\left(1 - \left(\frac{t^\star}{u}\right)^{\frac{1}{2+\delta_2}} \right)^{i - j - 1} \\
=& \frac{\left(t^\star/u\right)^{\frac{1}{2+\delta_2}}}{1 - \left(t^\star/u\right)^{\frac{1}{2+\delta_2}}} \sum_{i = j + 1}^\infty \frac{i - j}{j + \delta_2} \frac{\mathbb{P}\left(\xi_{\delta_2}^{j}(\tau^\star(u)) = i\right)}{i + \delta_2} \\
=& \frac{\left(t^\star/u\right)^{\frac{1}{2+\delta_2}}}{1 - \left(t^\star/u\right)^{\frac{1}{2+\delta_2}}}\left( \sum_{i = j + 1}^\infty \frac{\mathbb{P}\left(\xi_{\delta_2}^{j}(\tau^\star(u)) = i\right)}{j + \delta_2} -  \sum_{i = j + 1}^\infty \frac{\mathbb{P}\left(\xi_{\delta_2}^{j}(\tau^\star(u)) = i\right)}{i + \delta_2}\right) \\
=& \frac{\left(t^\star/u\right)^{\frac{1}{2+\delta_2}}}{1 - \left(t^\star/u\right)^{\frac{1}{2+\delta_2}}}\left(  \frac{\mathbb{P}\left(\xi_{\delta_2}^{j}(\tau^\star(u)) > j \right)}{j + \delta_2} + \frac{\mathbb{P}\left(\xi_{\delta_2}^{j}(\tau^\star(u)) = j\right)}{j + \delta_2} -  w_j\right) \\
=& \frac{\left(t^\star/u\right)^{\frac{1}{2+\delta_2}}}{1 - \left(t^\star/u\right)^{\frac{1}{2+\delta_2}}}\left( \frac{1}{j + \delta_2} -  w_j\right).
\end{align*}
Rearranging terms gives \eqref{eq:convex}. We further derive the following convenient upper and lower bounds for the sequence $\{w_{j}\}_{j = 1}^\infty$:
\begin{align}
\label{eq:wbounds}
\frac{(t^\star/u)^{\frac{1}{2+\delta_2}}}{j + \delta_2} (1 - w_j)< w_{j + 1} < \frac{(t^\star/u)^{\frac{1}{2+\delta_2}}}{j + \delta_2}.
\end{align}
In order to do so, see that
\begin{align*}
w_{j + 1} =& \sum_{i = j + 1}^\infty\frac{\mathbb{P}\left(\xi_{\delta_2}^{j + 1}(\tau^\star(u)) = i\right)}{i + \delta_2} \\
=& \sum_{k = j}^\infty\frac{\mathbb{P}\left(\xi_{\delta_2}^{j + 1}(\tau^\star(u)) = k + 1\right)}{k + 1 + \delta_2} \\
=& \sum_{k = j}^\infty\frac{1}{k + 1 + \delta_2} \frac{\Gamma(k + 1 + \delta_2)}{\Gamma(k - j + 1)\Gamma(j + 1 + \delta_2)} \left(\frac{t^\star}{u}\right)^{\frac{j + 1 + \delta_2}{2+\delta_2}}\left(1 - \left(\frac{t^\star}{u}\right)^{\frac{1}{2+\delta_2}} \right)^{k - j} \\
=& \frac{\left(\frac{t^\star}{u}\right)^{\frac{1}{2+\delta_2}}}{j + \delta_2}\sum_{k = j}^\infty\frac{k + \delta_2}{k + 1 + \delta_2} \mathbb{P}\left(\xi_{\delta_2}^{j}(\tau^\star(u)) = k\right) \\
=& \frac{\left(\frac{t^\star}{u}\right)^{\frac{1}{2+\delta_2}}}{j + \delta_2}\left(1- \sum_{k = j}^\infty\frac{1}{k + 1 + \delta_2} \mathbb{P}\left(\xi_{\delta_2}^{j}(\tau^\star(u)) = k\right)\right) \\
=& \frac{\left(\frac{t^\star}{u}\right)^{\frac{1}{2+\delta_2}}}{j + \delta_2}\left(1- \mathbb{E}\left[ \frac{1}{\xi_{\delta_2}^{j}(\tau^\star(u)) + 1 + \delta_2} \right]\right).
\end{align*}
Thus, clearly $w_{j + 1} < (t^\star/u)^{\frac{1}{2+\delta_2}}/(j + \delta_2)$. Also, 
\begin{align*}
\mathbb{E}\left[ \frac{1}{\xi_{\delta_2}^{j}(\tau^\star(u)) + 1 + \delta_2} \right] < \mathbb{E}\left[ \frac{1}{\xi_{\delta_2}^{j}(\tau^\star(u)) + \delta_2} \right] = w_j.
\end{align*}
Hence, \eqref{eq:wbounds} is proven. We now finally prove that $\{ w_{m} - w_{m + 1}\}_{m = 1}^\infty$ is a strictly decreasing sequence. From \eqref{eq:convex}, we have that
\begin{align*}
\frac{\left(t^\star/u\right)^{\frac{1}{2+\delta_2}}}{(j + \delta_2)(j + 1 + \delta_2)} =& \left(t^\star/u\right)^{\frac{1}{2+\delta_2}}(w_{j} - w_{j + 1}) \\
&+ \left(1 - \left(t^\star/u\right)^{\frac{1}{2+\delta_2}} \right)(w_{j + 1} - w_{j + 2}).
\end{align*}
Hence, it suffices to show that
\begin{align*}
w_{j + 1} - w_{j + 2} < \frac{\left(t^\star/u\right)^{\frac{1}{2+\delta_2}}}{(j + \delta_2)(j + 1 + \delta_2)}. 
\end{align*}
We may again apply \eqref{eq:convex} to see that
\begin{align*}
w_{j + 1} - w_{j + 2} =& \frac{1}{j + 1 + \delta_2} - \left(u/t^\star\right)^{\frac{1}{2+\delta_2}}w_{j + 2}, \\
\intertext{and applying the lower bound in \eqref{eq:wbounds} gives that }
<& \frac{1}{j + 1 + \delta_2} - \frac{1}{j + 1 + \delta_2}(1 - w_{j + 1}), \\ 
=& \frac{w_{j + 1}}{j + 1 + \delta_2},
\intertext{while applying the upper bound in \eqref{eq:wbounds} gives that } 
<& \frac{\left(t^\star/u\right)^{\frac{1}{2+\delta_2}}}{(j + \delta_2)(j + 1 + \delta_2)}.
\end{align*} 
Hence, we have prove that the sequence $\{ w_{m} - w_{m + 1}\}_{m = 1}^\infty$ is strictly decreasing. This allows us to apply Lemma 3 \cite{gao2017asymptotic} to state that if $\delta_1 < \delta_2$
\begin{align*}
\sum_{m = 1}^\infty p_{>m}(\delta_1)(w_{m} - w_{m + 1}) < \sum_{m = 1}^\infty p_{>m}(\delta_2)(w_{m} - w_{m + 1}),
\end{align*}
and if $\delta_2 < \delta_1$
\begin{align*}
\sum_{m = 21}^\infty p_{>m}(\delta_1)(w_{m} - w_{m + 1}) > \sum_{m = 1}^\infty p_{>m}(\delta_2)(w_{m} - w_{m + 1}).
\end{align*}
This implies that if $\delta_1 < \delta_2$
\begin{align*}
\sum_{j = 1}^\infty p_{j}(\delta_1)w_j > \sum_{j = 1}^\infty p_{j}(\delta_2)w_j,
\end{align*}
and if $\delta_1 > \delta_2$
\begin{align*}
\sum_{j = 1}^\infty p_{j}(\delta_1)w_j < \sum_{j = 1}^\infty p_{j}(\delta_2)w_j.
\end{align*}
Hence, the proof is complete.
\end{proof}

\subsection{Proof of Theorem~\ref{thm:consistency}}\label{subsec:pf_consistency}

\begin{proof}
To prove \eqref{eq:est}, we show that for any $\tau > 0$ and for some $\kappa > 0$
\begin{align}\label{eq:est_claim}
\mathbb{P}(-2 \log \Lambda_{nt^\star} > \sup_{s \in [\gamma, t^\star - \tau]} -2 \log \Lambda_{ns} + \kappa ) \to 1, \qquad \text{as } n\to\infty.
\end{align}
The proof for $s \in [t^\star + \tau, 1- \gamma]$ is similar and hence we omit it. Akin to the proof of Theorem \ref{thm:nullhyp}, we write
\begin{align*}
&-2 \log \Lambda_{nt^\star} + 2 \log \Lambda_{ns} \\
&=  2 \ell_{(0, nt^\star]}(\hat{\delta}_{(0, nt^\star]}) + 2 \ell_{(nt^\star, n]}(\hat{\delta}_{(nt^\star, n]}) - 2 \ell_{(0, ns]}(\hat{\delta}_{(0, ns]}) - 2 \ell_{(ns, n]}(\hat{\delta}_{(ns, n]}),\\
\intertext{and since $\ell_{(ns, n]}(\hat{\delta}_{(ns, n]})=\ell_{(ns, nt^\star]}(\hat{\delta}_{(ns, n]})+\ell_{(nt^\star, n]}(\hat{\delta}_{(ns, n]})$, and
$\ell_{(0, nt^\star]}(\hat{\delta}_{(0, nt^\star]}) = \ell_{(0, ns]}(\hat{\delta}_{(0, nt^\star]})+\ell_{(ns, nt^\star]}(\hat{\delta}_{(0, nt^\star]})$, we have}
&= 2 \left(\ell_{(0, ns]}(\hat{\delta}_{(0, nt^\star]}) - \ell_{(0, ns]}(\hat{\delta}_{(0, ns]}) \right) + 2 \left( \ell_{(ns, nt^\star]}(\hat{\delta}_{(0, nt^\star]}) - \ell_{(ns, nt^\star]}(\hat{\delta}_{(ns, n]}) \right) \\
& \ \ \  + 2 \left( \ell_{(nt^\star, n]}(\hat{\delta}_{(nt^\star, n]}) - \ell_{(nt^\star, n]}(\hat{\delta}_{(ns, n]}) \right) \\
&= 2 \left(\ell_{(0, ns]}(\hat{\delta}_{(0, nt^\star]}) - \ell_{(0, ns]}(\hat{\delta}_{(0, ns]}) \right) + 2 \left( \ell_{(ns, nt^\star]}(\hat{\delta}_{(ns, nt^\star]}) - \ell_{(ns, nt^\star]}(\hat{\delta}_{(ns, n]}) \right) \\
& \ \ \ + 2 \left(  \ell_{(ns, nt^\star]}(\hat{\delta}_{(0, nt^\star]}) - \ell_{(ns, nt^\star]}(\hat{\delta}_{(ns, nt^\star]}) \right) + 2 \left( \ell_{(nt^\star, n]}(\hat{\delta}_{(nt^\star, n]}) - \ell_{(nt^\star, n]}(\hat{\delta}_{(ns, n]}) \right); \\
\intertext{and since $\hat{\delta}_{(ns, nt^\star]}$ is a maximizer of $\ell_{(ns, nt^\star]}(\cdot)$, then we have the following lower bound:}
&\geq 2 \left( \ell_{(0, ns]}(\hat{\delta}_{(0, nt^\star]})  - \ell_{(0, ns]}(\hat{\delta}_{(0, ns]}) \right) + 2 \left( \ell_{(ns, nt^\star]}(\hat{\delta}_{(0, nt^\star]}) - \ell_{(ns, nt^\star]}(\hat{\delta}_{(ns, nt^\star]}) \right)\\
& \ \ \ + 2 \left( \ell_{(nt^\star, n]}(\hat{\delta}_{(nt^\star, n]}) - \ell_{(nt^\star, n]}(\hat{\delta}_{(ns, n]}) \right) \\
&:= A_1(n, s) + A_2(n, s)  + A_3(n, s),
\end{align*}
Invoking results in Theorem \ref{thm:nullhyp}, we see that $A_1(n, s)$ and $A_2(n, s)$ converge weakly to processes that are finite with probability $1$. 

We next claim that $\sup_{s \in [\gamma, t^\star - \tau]} A_3(n, s) \xrightarrow{p} \infty$, i.e.
\begin{align*}
\lim_{n\to\infty}\mathbb{P}\left(\sup_{s \in [\gamma, t^\star - \tau]} A_3(n, s) > L \right) = 1,
\end{align*}
for all $L > 0$.
The mean value theorem expansion allows us to write $A_3(n, s)$ as
\begin{equation}
\label{eq:MVT}
2 \left( \ell_{(nt^\star, n]}(\hat{\delta}_{(nt^\star, n]}) - \ell_{(nt^\star, n]}(\hat{\delta}_{(ns, n]}) \right) =  2u_{(nt^\star, n]}(\delta^{\star}_n(s))\left(\hat{\delta}_{(nt^\star, n]} - \hat{\delta}_{(ns, n]} \right),
\end{equation}
for some $\delta^{\star}_n(s)$ between $\hat{\delta}_{(ns, n]}$ and $\hat{\delta}_{(nt^\star, n]}$. Since $\hat{\delta}_{(nt^\star, n]}(\cdot)$ is maximized by $\hat{\delta}_{(nt^\star, n]}$, we can write 
\begin{align*}
 2 &\left( \ell_{(nt^\star, n]}(\hat{\delta}_{(nt^\star, n]}) - \inf_{s \in [\gamma, t^\star - \tau]} \ell_{(nt^\star, n]}(\hat{\delta}_{(ns, n]}) \right)\\
 & = \sup_{s \in [\gamma, t^\star - \tau]} 2\left| u_{(nt^\star, n]}(\delta^{\star}_n(s)) \right| \cdot \left| \hat{\delta}_{(nt^\star, n]} - \hat{\delta}_{(ns, n]} \right|
\end{align*}
Our proof consists of two parts, and we will show that as $n\to\infty$,
\begin{enumerate}
\item[(1)] $\mathbb{P}\left( \inf_{s \in [\gamma, t^\star - \tau]} \left| \hat{\delta}_{(nt^\star, n]} - \hat{\delta}_{(ns, n]} \right| > \epsilon \right) \rightarrow 1$;
\item[(2)] $\sup_{s \in [\gamma, t^\star - \tau]} \left| u_{(nt^\star, n]}(\delta^{\star}_n(s)) \right| \xrightarrow{p} \infty$.
\end{enumerate}

We now focus on (1) and consider the term $\left| \hat{\delta}_{(nt^\star, n]} - \hat{\delta}_{(ns, n]} \right|$. 
Note that $u_{(ns, n]}(\lambda) =  u_{(ns, nt^\star]}(\lambda) + u_{(nt^\star, n]}(\lambda)$, and  Lemmas \ref{lem:score_prime_asymp} and \ref{lem:score_asymp_chng} give
\begin{align}
&\sup_{\lambda \in [\eta, K]}\left|n^{-1}u_{(ns, nt^\star]}(\lambda) - (t^\star - s)U(\lambda; \delta_1) \right| \xrightarrow{p} 0,\label{eq:unifconvergence1}\\
& \sup_{\lambda \in [\eta, K]}\left|n^{-1}u_{(nt^\star, n]}(\lambda) - (U^\star_1(\lambda) - U^\star_{t^\star}(\lambda)) \right| \xrightarrow{p} 0 \label{eq:unifconvergence2}
\end{align}
Hence
\begin{equation}
\label{eq:unifconv}
\sup_{\lambda \in [\eta, K]} \left|n^{-1}u_{(ns, n]}(\lambda) - V(\lambda)\right| \xrightarrow{p} 0,
\end{equation}
where we set $V(\lambda) := (t^\star - s)U(\lambda; \delta_1) + U^\star_1(\lambda) - U^\star_{t^\star}(\lambda)$. Also note that $U(\lambda; \delta_1)$ and $ U^\star_1(\lambda) - U^\star_{t^\star}(\lambda)$ have unique zeros at $\delta_1$ and $\delta_2$, respectively, and are positive before and negative after their zeros (see Lemma 4 in \cite{gao2017asymptotic} and Lemma \ref{lem:unique_zero}).

In particular, $|V(\delta_2)| = |(t^\star - s)U(\delta_2; \delta_1) + U^\star_1(\delta_2) - U^\star_{t^\star}(\delta_2)| =  |(t^\star - s)U_1(\delta_2)| > 0$ since $U(\lambda; \delta_1)$ has a unique zero at $\delta_1$ and is positive beforehand. Fix $0 < \xi < |V(\delta_2)|$, and the continuity of $V(\lambda)$ thus gives the existence of a $\epsilon > 0$ such that 
\begin{equation*}
\inf_{\lambda : |\lambda - \delta_2| < 2\epsilon} \left| V(\lambda) \right| > \xi.
\end{equation*}
Further, since $n^{-1}u_{(ns, n]}(\hat{\delta}_{(ns, n]})= 0$,
\begin{equation}
\label{eq:U3}
\left| V(\hat{\delta}_{(ns, n]}) \right|  =  \left| V(\hat{\delta}_{(ns, n]}) -  n^{-1}u_{(ns, n]}(\hat{\delta}_{(ns, n]}) \right| \leq \sup_{\lambda \in [\epsilon, K]} \left| \frac{1}{n}u_{(ns, n]}(\lambda)  - V(\lambda) \right|.
\end{equation}
Therefore, we conclude from \eqref{eq:unifconv} that with probability tending towards 1,
\begin{align*}
\left|  V(\hat{\delta}_{(ns, n]}) \right| \leq  \sup_{\lambda \in [\epsilon, K]} \left| \frac{1}{n}u_{(ns, n]}(\lambda)  - V(\lambda) \right| \leq \xi < \inf_{\lambda : |\lambda - \delta_2| < 2\epsilon} \left| V(\lambda) \right|,
\end{align*}
which further implies $\mathbb{P}(|\hat{\delta}_{(ns, n]} - \delta_2 | > 2\epsilon) \rightarrow 1$. Meanwhile, by the consistency of $\hat{\delta}_{(nt^\star, n]}$, we have
$$\mathbb{P}(|\hat{\delta}_{(nt^\star, n]} - \delta_2| < \epsilon) \rightarrow 1.$$ Hence,
we see that as $n\to\infty$,
\begin{equation}
\label{eq:difference}
\mathbb{P}\left( \inf_{s \in [\gamma, t^\star - \tau]}  \left| \hat{\delta}_{(nt^\star, n]} - \hat{\delta}_{(ns, n]} \right|  > \epsilon \right) \rightarrow 1.
\end{equation}

We next show that $\sup_{s \in [\gamma, t^\star - \tau]}  |u_{(nt^\star, n]}(\delta^{\star}_n(s))| \xrightarrow{p} \infty$. Let $\mathcal{B}_{r}(x)$ denote the open ball of radius $r$ centered at $x$. Note that as $n \rightarrow \infty$, $\delta^{\star}_n(s) \notin \mathcal{B}_{\kappa}(\delta_2)$ with high probability for some $\kappa > 0$ since $\hat{\delta}_{(nt^\star, n]}$ and $\hat{\delta}_{(ns, n]}$ are well-separated with probability tending towards $1$ and $\ell_{(nt^\star, n]}(\cdot)$ is non-linear and maximized at $\hat{\delta}_{(nt^\star, n]}$. Further,

\begin{align*}
\left|n^{-1}u_{(nt^\star, n]}(\delta^{\star}_n(s))\right| \geq& \left|U^\star_1(\delta^{\star}_n(s)) - U^\star_{t^\star}(\delta^{\star}_n(s))\right| \\
& - \sup_{\lambda \in [\eta, K]} \left|n^{-1}u_{(nt^\star, n]}(\lambda) - (U^\star_1(\lambda) - U^\star_{t^\star}(\lambda)) \right| \\
\geq& \inf_{\lambda \notin B_{\kappa}(\delta_2)} \left|U^\star_1(\lambda) - U^\star_{t^\star}(\lambda)\right| \\
& - \sup_{\lambda \in [\eta, K]} \left|n^{-1}u_{(nt^\star, n]}(\lambda) - (U^\star_1(\lambda) - U^\star_{t^\star}(\lambda)) \right|.
\end{align*}

The first term is nonzero since $U^\star_1(\cdot) - U^\star_{t^\star}(\cdot)$ has a unique zero at $\delta_2$. The second term can be make arbitrarily small by by \eqref{eq:unifconvergence1} and \eqref{eq:unifconvergence2}. Hence
\begin{equation}\label{eq:uprime}
\sup_{s \in [\gamma, t^\star - \tau]} \left| u_{(nt^\star, n]}(\delta^{\star}_n(s)) \right| \xrightarrow{p} \infty.
\end{equation}
Combining \eqref{eq:uprime} with \eqref{eq:difference} gives that
$-2 \log \Lambda_{nt^\star} + \inf_{s \in [\gamma, t^\star - \tau]} 2 \log \Lambda_{ns} \xrightarrow{p} \infty$, which implies the existence of a $\kappa > 0$ such that
\begin{align*}
\mathbb{P}\left( -2 \log \Lambda_{nt^\star} > \sup_{s \in [\gamma, t^\star - \tau]} -2 \log \Lambda_{ns} + \kappa \right) \rightarrow 1.
\end{align*}
Therefore, as $n \rightarrow \infty$, the maximum becomes well-separated, and occurs at $t^\star$ so that we have
\begin{align*}
\hat{t}_n = n^{-1} \underset{m = \floor{n\gamma}, \dots, \floor{n(1 -\gamma)}}{\argmax} \Lambda_m \xrightarrow{p} t^\star.
\end{align*}
\end{proof}

\subsection{Proof of Theorem~\ref{thm:nullhyp2}}\label{subsec:pf_null2}

\begin{proof}
We first observe that via a Taylor expansion of $u_{(0, nt]}(\lambda)$ around $\hat{\delta}_{(0, nt]}$
\begin{align*}
n^{-1/2}u_{(0, nt]}(\hat{\delta}_n) &= n^{-1/2}u_{(0, nt]}(\hat{\delta}_{(0, nt]}) + n^{-1/2}u'_{(0, nt]}(\delta^\star_n)(\hat{\delta}_n - \hat{\delta}_{(0, nt]}) \\
&= n^{-1/2}u'_{(0, nt]}(\delta^\star_n)(\hat{\delta}_n - \hat{\delta}_{(0, nt]}).
\end{align*}
where $\delta^\star_n$ lies between $\hat{\delta}_n$ and $\hat{\delta}_{(0, nt]}$. Similarly to the argument made in Lemma \ref{lem:MLE}, 
\begin{align*}
-n^{-1}u'_{(0, nt]}(\delta^\star_n) \xrightarrow{p} tI(\delta; \delta) \qquad \text{in }D[0,1],
\end{align*}
and by \eqref{eq:MLEfunctional2},
\begin{align*}
t I(\delta; \delta) \cdot \sqrt{n} \left(\hat{\delta}_n -  \hat{\delta}_{(0, nt]} \right) \Rightarrow tW\left( I(\delta; \delta)\right) - W\left( tI(\delta; \delta) \right)\qquad \text{in }D[\gamma,1],
\end{align*}
which implies
\begin{align*}
n^{-1}u^2_{(0, nt]}(\hat{\delta}_n) \Rightarrow  \left( tW\left( I(\delta; \delta)\right) - W\left( tI(\delta; \delta) \right) \right)^2 \qquad \text{in }D[\gamma,1].
\end{align*}
Further, Lemma \ref{lem:score_prime_asymp} gives that in $D[0,1]$, 
\begin{align*}
-n^{-1}u'_{(0, nt]}(\hat{\delta}_n) &\xrightarrow{p} tI(\delta; \delta),\\
-n^{-1}u'_{(nt, n]}(\hat{\delta}_n) &\xrightarrow{p} (1 - t)I(\delta; \delta).
\end{align*}
Define the Brownian bridge $\displaystyle B(t):= I^{-1/2}(\delta; \delta) \left(W(tI(\delta; \delta)) - tW(I(\delta; \delta)) \right)$. Combining the previous three convergences thus gives
\begin{align*}
S_{nt} = -u^2_{(0, nt]}(\hat{\delta}_n)\left( \frac{1}{u'_{(0, nt]}(\hat{\delta}_n)} + \frac{1}{u'_{(nt, n]}(\hat{\delta}_n)} \right) \Rightarrow B^2(t)\left(\frac{1}{t} + \frac{1}{1 - t} \right) = \frac{B^2(t)}{t(1 - t)}
\end{align*}
in $D[\gamma,1-\gamma]$.
Then applying the continuous mapping theorem gives
\begin{align*}
\sup_{t \in [\gamma, 1 - \gamma]} S_{nt} \Rightarrow \sup_{t \in [\gamma, 1 - \gamma]} \frac{B^2(t)}{t(1 - t)} \qquad \text{in }\mathbb{R}.
\end{align*}
\end{proof}





\end{document}